\documentclass[11pt]{article}
\usepackage[dvips]{graphicx}
\usepackage{amsmath}
\usepackage{amssymb}
\usepackage{amsthm}
\usepackage{kikuuwe}
\usepackage{eqnarraycompact}
\usepackage{color}
\usepackage[font=small]{caption}
\iffalse
\usepackage[round]{natbib}
\renewcommand\cite\citep
\else
\usepackage{cite}
\fi
\newtheorem{theorem}{Theorem}
\newtheorem{proposition}{Proposition}
\newtheorem{remark}{Remark}

\newcommand\gsgn{\mathrm{gsgn}}
\newcommand\diag{\mathrm{diag}}
\newcommand\proj{\mathrm{proj}}
\newcommand\fb{\bar{f}}
\newcommand\uh{\hat{u}}

\usepackage{footmisc}
\newcommand{\customfootnotetext}[2]{{
  \renewcommand{\thefootnote}{#1}
  \footnotetext[0]{#2}}}

\begin{document}

\title{Nonsmooth Quasistatic Modeling of Hydraulic Actuators}

\author{Ryo Kikuuwe\footnotemark[1]{\ }\footnotemark[4], Tomofumi Okada\footnotemark[2]{\ }\footnotemark[3], Hideo Yoshihara\footnotemark[2], \\ Takayuki Doi\footnotemark[3], Takao Nanjo\footnotemark[3]{ } {\it and} Koji Yamashita\footnotemark[3]%
}
\maketitle%

\customfootnotetext{$\star$}{Graduate School of Advanced Science and Engineering, Hiroshima University, 1-4-1 Kagamiyama, Higashi-Hiroshima, Hiroshima 739-8527, Japan. e-mail: kikuuwe@ieee.org}
\customfootnotetext{\dag}{Kobelco Construction Machinery Dream-Driven Co-Creation Research Center, Hiroshima University, 1-4-1 Kagamiyama, Higashi-Hiroshima, Hiroshima 739-8527, Japan. }
\customfootnotetext{\ddag}{Kobelco Construction Machinery Co., Ltd., 2-2-1 Itsukaichiko, Saeki-ku, Hiroshima 731-5161, Japan.}
\customfootnotetext{$\mathsection$}{Corresponding Author}

\begin{abstract}
This article presents a quasistatic model of a hydraulic actuator driven by a four-valve independent metering circuit. The presented model describes the quasistatic balance between the velocity and force and that between the flowrate and the pressure. In such balanced states, the pressure difference across each valve determines the oil flowrate through the valve, the oil flowrate into the actuator determines the velocity of the actuator, and the pressures in the actuator chambers are algebraically related to the external force acting on the actuator. Based on these relations, we derive a set of quasistatic representations, which analytically relates the control valve openings, the actuator velocity, and the external force. This analytical expression is written with nonsmooth functions, of which the return values are set-valued instead of single-valued. We also show a method of incorporating the obtained nonsmooth quasistatic model into multibody simulators, in which a virtual viscoelastic element is used to mimic transient responses. In addition, the proposed model is extended to include a regeneration pipeline and to deal with a collection of actuators driven by a single pump.
\end{abstract}
\begin{quote}\small
\end{quote}

\section{Introduction}

\newcommand\tmpfootnoteA{\footnote{We use the term `quasistatic' because the model does not involve the pressure dynamics but involves the motion of the actuator and the oil.
}}

Control technology for construction machines requires continuing research and development for future applications, such as remote and semi-automatic operation. The productivity of research and development heavily depends on simulation techniques. In particular, hydraulic actuators and hydraulic circuits are important components of construction machines, of which the physical behaviors need to be appropriately modeled in simulators. The behavior of a hydraulic actuator is highly involving, depending on the oil supply from the pump, the external forces acting on the actuator, and the states and the characteristics of many valves in the circuit. 

\par

Many of the previous studies on the modeling of hydraulic systems assume that the pressure is governed by first-order dynamics. The circuit is often divided into several oil volumes, such as those in the circuit pipelines and the actuator chambers. In each volume, the rate-of-change of the pressure is determined by the oil flowrates in and out of the volume. Such an approach, sometimes referred to as a lumped fluid approach, has been employed for the controller design \cite{YaoJ_2014_Hydraulic,Christofori_2015_Model,Ruderman_2017_Full,Destro_2018_Valves} and for simulation purposes~\cite{Ylinen_2014_Hydraulic,Sakai_2018_Nondim}. Coupling of hydraulic systems and multibody systems have also been studied~\cite{Ylinen_2014_Hydraulic,Rahikainen_2018_LumpedFluid,Rahikainen_2018_Efficient,Rahikainen_2020_Cosimulation}.

\par

In the conventional model of the pressure dynamics, the rate-of-change of the pressure is proportional to the bulk modulus divided by the volume of the oil. The bulk modulus of the oil is usually high, and the oil volume may be small when one needs to deal with a small segment in the pipes and when the piston approaches either end of the cylinder. Therefore, the differential equations representing the pressure dynamics can become numerically stiff, demanding a small timestep size and a high computational cost for use in simulation. Some researchers \cite{WangL_2012_Singular,Rahikainen_2018_Efficient,KianiOshtorjani_2019_Hydraulic} applied the singular perturbation theory to avoid the numerical stiffness of the governing differential equations of the pressure. In the singular perturbation approach, the pressure dynamics is assumed to be so fast that the steady-state pressures are quickly achieved. Following this notion, the pressure rate-of-change in the governing differential equation is replaced by zero, and the pressure dynamics is converted into an algebraic constraint determining the steady-state pressure. Kiani Oshtorjani~et~al.~\cite{KianiOshtorjani_2019_Hydraulic} have explored the applicability of this scheme to distinguish which pressure rates-of-change can or cannot be zeroed based on the exponential stability of the original differential equation of the pressure dynamics.

\par

\begin{figure}[t!]\begin{center}
\includegraphics[scale=1.2]{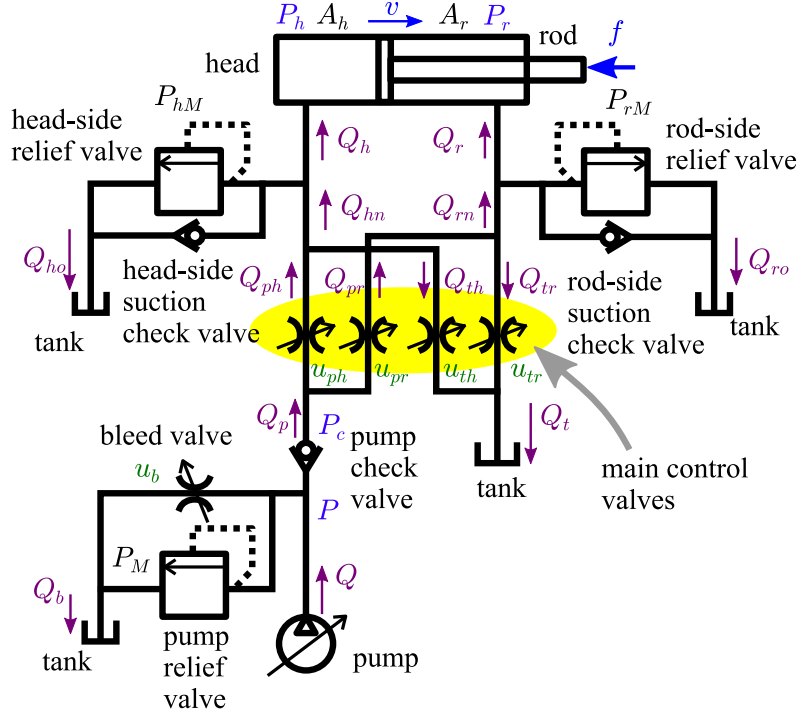}\\[-4pt]
\caption{Hydraulic actuator and its circuit.}\label{fg:circuit}
\end{center}\end{figure}

This article proposes a computationally-efficient modeling scheme for hydraulic actuators, particularly focusing on the circuit structure shown in Fig.~\ref{fg:circuit}. This approach focuses on the quasistatic\tmpfootnoteA{} balances established at the steady state, at which the pressure difference across each valve determines the oil flowrate through the valve, the oil flowrate into the actuator determines the actuator's velocity, the pressures in the actuator chambers determine the force generated by the actuator, and the generated force equals the external force acting on the actuator. This article elaborately derives an analytical expression of the quasistatic relation between the actuator velocity and the external force from the quasistatic representations of all valves in the circuit. One important feature of the presented expression is that it is nonsmooth, allowing for the set-valuedness of the pressure especially when the valves are closed or the actuator is stopped.


\par

The presented approach may be viewed as a full application of the singular perturbation theory in the sense that it provides an approximate solution by neglecting the first-order pressure dynamics. In this approach, the pressures and flowrates are constrained algebraically, instead of through the dynamics. In contrast to previous works \cite{WangL_2012_Singular,Rahikainen_2018_Efficient,KianiOshtorjani_2019_Hydraulic}, the presented nonsmooth model allows for multiple steady states by involving the set-valuedness, and describes the quasistatic balance between the chamber pressures and the external force, to which the system should converge in the steady state.

\par

In addition to the quasistatic model of Fig.~\ref{fg:circuit}, this article presents a method to use the model in multibody dynamics simulation. In simulations, the actuator model is connected with other mechanical components through a virtual viscoelastic element to deal with the set-valuedness of the actuator model. This approach is what Kikuuwe~\cite{Kikuuwe_2018_VSS} has been referring to as a {\it differential-algebraic relaxation}, which has been applied to many nonsmooth problems by Kikuuwe and his colleagues~\cite{Kikuuwe_2006_TRO,Kikuuwe_2010_TRO,Kikuuwe_2014_TMECH,Xiong_2013_JAMH,Xiong_2013_Tribo,Kikuuwe_2019_MPE}. Some extensions toward more complicated circuit structures are also presented.

\par

This article is organized as follows. Section~\ref{sc:prm} shows mathematical preliminaries, including definitions of relevant functions and some theorems and propositions. Section~\ref{sc:stt} constructs a nonsmooth quasistatic model of the hydraulic circuit of Fig.~\ref{fg:circuit}, which describes the algebraic relation among the valve openings, the rod velocity, and the external force. Section~\ref{sc:dyn} presents an approach to incorporate the quasistatic model into multibody simulators. Section~\ref{sc:reg} presents an extended model including a regeneration pipeline, and Section~\ref{sc:mul} presents a model including multiple actuators driven by a single pump. Section~\ref{sc:ccl} provides some concluding remarks.

\section{Mathematical Preliminary}\label{sc:prm}
\subsection{Some Nonsmooth Functions}

In this article, $\bbR$ denotes the set of all real numbers. This article extensively uses mathematical notations of the nonsmooth system theory, which involves set-valued functions. We use the following set-valued functions:
\beqa
\calN_{[A,B]}(x) &\define& \lb\{\barr{ll}
 [0,\infty) & \mbox{if}\ x= B \\
0 & \mbox{if}\ x\in(A,B) \\
 (-\infty,0] & \mbox{if}\ x=A \\
\emptyset & \mbox{otherwise} \earr\rb.  \dfeq{defcalN}
\\
\gsgn(a,x,b)&\define&\lb\{\barr{ll} b & \mbox{if}\ x>0 \\ \mathrm{cl}\{a,b\} &\mbox{if}\ x=0 \\ a &\mbox{if}\ x<0. \earr\rb.
\eeqa
The definition \eq{defcalN} assumes $A<B$. Here, $\mathrm{cl}\{a,b\}$ stands for the convex closure of the set $\{a,b\}$, being the closed set $[a,b]$ if $a\leq b$ and $[b,a]$ if $b\leq a$. The function $\calN_{\calF}(x)$ is referred to as the normal cone of the set $\calF$ at the point $x$. The function $\gsgn$ can be seen as a generalized version of the set-valued signum function. With the normal cone $\calN$, the following relation holds true:
\beqa
 0\leq x \perp y\geq 0  \ \iff\ 
 x\in-\calN_{[0,\infty)}(y)\ \iff\  
 y\in-\calN_{[0,\infty)}(x) .
\eeqa
Each of the above three expressions means that $x$ and $y$ are non-positive scalars at least one of which is zero. This relation is convenient to describe the flowrate-pressure relation at check valves. In addition, with $A<B$, the following relation holds true:
\beqa
y \in\calN_{[A,B]}(x) \ \iff\ x\in\gsgn(A,y,B). \dfeq{calNgsgn}
\eeqa

\par

The following function represents the projection onto a closed set:
\beqa
\proj_{[A,B]}(x) &\define& \max(A,\min(B,x)) \dfeq{defproj}
\eeqa
where $A<B$. The normal cone and the projection have the following relation:
\beqa
a-x\in\calN_{\calA}(x) \ \iff\ x=\proj_{\calA}(a), \dfeq{normproj}
\eeqa
which has been shown in previuos publications \cite[Proposition~2]{Brogliato_2006_Equivalence}\cite[Section~A.3]{Acary_2008_Numerical}. This article uses the following theorem:
\begin{theorem}\label{thm:mainp2}
Let $x\in\bbR$ and let $\calA\subset\bbR$ be a closed subset of $\bbR$. Let $f:\bbR\to\bbR$ be a strictly decreasing function. Let $x_f\in\bbR$ satisfy $f(x_f)=0$. Then, the following statement holds true: 
\beqa
f(x)\in\calN_{\calA}(x)\ \iff\ x=\proj_{\calA}(x_f).
\eeqa
\end{theorem}
\begin{proof}
Because $f$ is strictly decreasing, $x<x_f \iff f(x)<0$, $x>x_f \iff f(x)>0$, and $x=x_f \iff f(x)=0$ are satisfied. This means that $f(x)\in\calN_{\calA}(x) \iff x_f-x \in\calN_{\calA}(x)$ and \eq{normproj} implies that it is equivalent to $x=\proj_{\calA}(x_f)$.
\end{proof}

\par

\subsection{Some Smooth Functions}\label{ss:ssf}

The following single-valued functions are used in the article:
\beqa
\calS(x) &\define& \sgn(x)x^2 \\
\calR(x) &\define& \sgn(x)\sqrt{|x|}\\
\psi(u_1,u_2)&\define& \lb\{\barr{ll} 0 & \mbox{if }\ u_1=u_2=0 \\ \dfrac{u_1^2u_2^2}{u_1^2+u_2^2} &\mbox{otherwise.} \earr\rb.
\eeqa
The functions $\calR$ and $\calS$ are strictly increasing continuous functions satisfying $\calR(\calS(x))=\calS(\calR(x))=x$. The function $\psi$ is differentiable everywhere.

\par

Section~\ref{sc:dyn} will use the following functions:
\beqa
\Phi_A(b,c,a)&\define& -\sgn(c)\dfrac{\sqrt{a^2b^2+4a|c|}-ab}{2}
\\
\Phi_B(b,c,a_0,a_1,x_1)&\define& \lb\{\barr{lr}
-\dfrac{\sqrt{a_0^2(a_1b+2x_1)^2+4a_0(a_0+a_1)(a_1c-x_1^2)}-a_0(a_1b+2x_1)}{2(a_0+a_1)}
\hspace{-12cm}&\\[12pt]&\mbox{if }
\lb( -a_0(bx+c)\leq \calS(x_1)\leq 0\rb)\,\vee\,\lb(0 \leq \calS(x_1)\leq a_1c \rb)
\\[12pt]
\dfrac{\sqrt{a_0^2(a_1b-2x_1)^2-4a_0(a_0+a_1)(a_1c+x_1^2)}-a_0(a_1b-2x_1)}{2(a_0+a_1)}
\hspace{-8cm}&\\[12pt]&\mbox{if }
\lb(0\leq \calS(x_1)\leq -a_0(bx+c) \rb)\,\vee\,\lb(a_1c \leq \calS(x_1)\leq0 \rb)
\\[12pt]
-\dfrac{2\sqrt{a_0}(a_1c+x_1^2)}{\sqrt{a_0}(a_1b-2x_1)+\sqrt{a_0(a_1b-2x_1)^2+4(a_1-a_0)(a_1c+x_1^2)}}
\hspace{-12cm}&\\[12pt]&\mbox{if }
\calS(x_1)\leq\min(0,-a_0(bx+c),a_1c)  
\\[12pt]
\dfrac{2\sqrt{a_0}(x_1^2-a_1c)}{\sqrt{a_0}(a_1b+2x_1)+\sqrt{a_0(a_1b+2x_1)^2+4(a_1-a_0)(x_1^2-a_1c)}}
\hspace{-16cm}&\\[12pt]&\mbox{if }
\max(0,-a_0(bx_1+c),a_1c)\leq\calS(x_1).
\earr\rb.
\eeqa
With these functions, we have the following propositions:
\begin{proposition}\label{prp:PhiA}
Let $a\geq 0$ and $b>0$. Then, the following statement holds true:
\beqa
x=\Phi_A(b,c,a)\ \iff\ \lim_{\tilde{a}\to a}\dfrac{\calS(x)}{\tilde{a}}+bx+c=0.
\eeqa
\end{proposition}
\begin{proposition}\label{prp:PhiB}
Let $a_0\geq 0$, $a_1\geq 0$, and $b>0$, and let $a_0^2+a_1^2>0$ be satisfied. Then, the following statement holds true:
\beqa
x=\Phi_B(b,c,a_0,a_1,x_1)\ \iff\ \lim_{\tilde{a}_0\to a_0}\dfrac{\calS(x)}{\tilde{a}_0}+\lim_{\tilde{a}_1\to a_1}\dfrac{\calS(x-x_1)}{\tilde{a}_1}+bx+c= 0.
\eeqa
\end{proposition}
The proofs of these propositions can be obtained through tedious but straightforward derivations.

\par

We also use the following theorems, of which the proofs are rather trivial:
\begin{theorem}\label{thm:fgswitch}
Let $f:\bbR\to\bbR$ and $g:\bbR\to\bbR$ be strictly decreasing functions and let $g(0)\leq f(0)$ be satisfied. Let $x_f$ and $x_g$ satisfy $b=f(x_f)$ and $b=g(x_g)$ where $b\in\bbR$. Then, the following statement holds true:
\beqa
b\in\gsgn(f(x),x,g(x)) \ \iff\ 
x = \lb\{\barr{ll} 
x_f &\mbox{if } b > f(0) \\
x_g &\mbox{if } b < g(0) \\
0   &\mbox{otherwise}. \\
\earr\rb.
\eeqa
\end{theorem}
\begin{theorem}\label{thm:main}
Let $f_1:\bbR\to\bbR$ and $f_2:\bbR\to\bbR$ be strictly decreasing functions. Let $x_i$ satisfy $f_i(x)=0$ for $i\in\{1,2\}$. Then, the following two statements hold true:
\beqa
\min(f_1(x),f_2(x))=0 &\iff& x = \min(x_1,x_2) \\
\max(f_1(x),f_2(x))=0 &\iff& x = \max(x_1,x_2).
\eeqa
\end{theorem}

\section{Nonsmooth Quasistatic Model}\label{sc:stt}

This section considers the hydraulic circuit illustrated in Fig.~\ref{fg:circuit}, which is a four-valve independent metering circuit to drive a double-acting hydraulic actuator. The actuator has two chambers separated by the piston, and the motion of the piston is extracted as the motion of the rod, which applies forces to external objects. We are interested in the quasistatic relation among the rod velocity $v$, positive when the rod is extending, the external force $f$, positive when it is compressing the rod, and the opening ratios of the valves. This hydraulic circuit is similar to those studied in, e.g., \cite{Tabor_2005_Four,Shenouda_2006_PhD,Shenouda_2008_Switching,Eriksson_2011_Metering,Choi_2015_Excavator}, where the quasistatic relations are also considered. Our main contribution lies in an elaborate analytical representation of the whole circuit, which is rather complicated than those in previous studies, using the nonsmooth formalism to deal with relief valves and check valves.

\par

This article uses the terminology for linear hydraulic actuators (i.e., hydraulic cylinders), but the presented approach is applicable also to rotary hydraulic actuators by replacing the velocity and the external force by the angular velocity and the torque, respectively.

\subsection{Quasistatic Relations}

In Fig.~\ref{fg:circuit}, $Q_*$ and $P_*$ denote the flowrates and the pressure at each point. The pump provides the flowrate $Q$ to the circuit via a pump check valve and two of the four main control valves. These control valves are connected to the head-side and the rod-side chambers of the actuator. The chambers are also connected to the tank with the zero pressure via the other two control valves. Each chamber of the actuator also connects to the tank through a parallel combination of a relief valve and a check valve, which are named as indicated in the figure. There is another control valve, referred to as a bleed valve, and it leads to the tank in parallel to a relief valve, referred to as a pump relief valve.

\par

\begin{figure}[t!]\begin{center}
\includegraphics[scale=1.1]{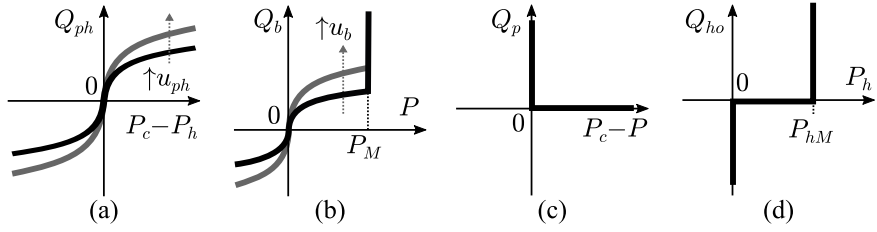}\\[-4pt]
\caption{Flowrate-pressure relations at: (a) one of main control valves \eq{qphdd}, (b) the bleed and the pump relief valves \eq{QboA}, (c) the pump check valve \eq{Qp}, and (d) the head-side relief and the suction check valves \eq{Qho}.  }\label{fg:graphs}
\end{center}\end{figure}

The degrees of the opening of the valves are represented by dimensionless variables $u_*\in[0,1]$ ($*\in\{ph,tr,pr,th,b\}$), which are the ratios of the valve opening areas to their maximum values. The control valves are manipulated by a controller that accepts external commands, which are given through, e.g., operation levers of an excavator. When the external command is to stop the actuator, all the four main valves are closed. When the external command is to move the actuator in the positive direction (i.e., extend the rod), both or either of $u_{ph}$ and $u_{tr}$ are set positive and $u_{th}$ and $u_{pr}$ are set zero. When the operator's command is to move the actuator in the negative direction (i.e, retract the rod), $u_{ph}$ and $u_{tr}$ are set zero and both or either of $u_{th}$ and $u_{pr}$ are set positive. Based on this idea, we assume that the vector $\bmu\define[u_{ph},u_{tr},u_{pr},u_{th},u_b]^T$ always belongs to either of the following three subsets: 
\begin{subequations}\dfeq{calUdef}
\beqa
\bmcalU_0 &\define& \{\bmu\in\bbR^5\,|\, u_1=u_2=u_3=u_4=0 \,\wedge\, u_5>0         \} \\
\bmcalU_+ &\define& \{\bmu\in\bbR^5\,|\, u_1^2+u_2^2>0     \,\wedge\, u_3=u_4=0     \,\wedge\, u_5\geq 0 \}\\
\bmcalU_- &\define& \{\bmu\in\bbR^5\,|\, u_1=u_2=0         \,\wedge\, u_3^2+u_4^2>0 \,\wedge\, u_5\geq 0 \} 
\eeqa
\end{subequations}
where $u_i$ ($i\in\{1,\cdots,5\}$) stands for the $i$th element of the vector $\bmu$. 

\par

Now, let us make an exhaustive list of algebraic relations among the pressures, the flowrates, the external force, and the actuator velocity at the steady state. First, according to the principle of mass conservation, one can see that the following relations hold true at junctions in the circuit:
\begin{subequations}\dfeq{Q}%
\beqa
Q_{hn}&=& Q_{ph}-Q_{th} \\
Q_{rn}&=& Q_{pr}-Q_{tr} \\
Q_h   &=& Q_{hn}-Q_{ho} \\
Q_r   &=& Q_{rn}-Q_{ro} \\
Q_t   &=& Q_{th}+Q_{tr} \\
Q_p   &=& Q_{ph}+Q_{pr} \\
Q_b   &=& Q-Q_p. 
\eeqa

\par

Second, let us focus on the control valves. As indicated in Fig.~\ref{fg:circuit}, $P_h$ and $P_r$ are the internal pressures of the head- and rod-side chambers, respectively, $P_c$ is the pressure at the check valve connected to the pump, and $P$ is the pressure at the outlet of the pump. According to the conventional orifice model \cite{Borutzky_2002_Orifice,Christofori_2015_Model}, we can assume that the following flowrate-pressure relations are satisfied:
\beqa
Q_{ph}& =& c_{ph} u_{ph}  \calR(P_c-P_h) \dfeq{qphdd} \\
Q_{th}& =& c_{th} u_{th}  \calR(P_h) \\
Q_{pr}& =& c_{pr} u_{pr}  \calR(P_c-P_r) \\
Q_{tr}& =& c_{tr} u_{tr}  \calR(P_r) \\
Q_b   &\in& c_b    u_b  \calR(P)+\calN_{(-\infty,P_M]}(P). \dfeq{QboA}
\eeqa
Fig.~\ref{fg:graphs}(a) illustrates the relation \eq{qphdd}. The coefficients $c_*$ are defined as $c_* = C_* a_* \sqrt{2/\rho}$ ($*\in\{ph,tr,pr,th,b\}$) where $C_*$ is the dimensionless coefficient named a discharge coefficient~\cite{Lichtarowicz_1965_Discharge}, which is typically around $0.6$ or $0.7$~\cite{YeY_2014_Valve,WuD_2002_Discharge}, $a_*$ is the maximum opening area (m$^2$) of the valve, and $\rho$ is the mass density (kg/m$^3$) of the oil. Equation \eq{QboA} represents the combined effect of the relief valve and the bleed valve, which limits the pump pressure up to $P_M$ as illustrated in Fig.~\ref{fg:graphs}(b).

\par

Third, let us consider the check valves and the relief valves. The check valve connected to the pump imposes the following constraint:
\beqa
Q_p&\in&-\calN_{[0,\infty)}(P_c-P), \dfeq{Qp}
\eeqa
which means that the flowrate $Q_p$ is zero when $P_c-P>0$, as illustrated in Fig.~\ref{fg:graphs}(c). The effects of the relief valves and the suction check valves connected to the chambers are written as follows:
\beqa
Q_{ho} &\in& \calN_{[0,P_{hM}]}(P_h)  \dfeq{Qho}\\
Q_{ro} &\in& \calN_{[0,P_{rM}]}(P_r).
\eeqa
Here, $P_{hM}$ and $P_{rM}$ are the pressure limits of the relief valves connected to the head- and rod-side chambers, respectively. These expressions mean that $P_h$ and $P_r$ are always in the ranges of $[0,P_{hM}]$ and $[0,P_{rM}]$, respectively, as illustrated in Fig.~\ref{fg:graphs}(d). When the pressure reaches the upper limit, the oil flows into the tank. When the pressure reaches zero, the oil is drawn into the chamber from the tank.

\par

Lastly, we discuss the actuator. Let $A_h$ and $A_r$ be the cross-sectional areas of the head- and rod-side chambers, respectively. Then, at the steady state where the rod inertia can be neglected, the constraints imposed by the actuator can be written as follows:
\beqa
v&=&Q_h/A_h \\
v&=&-Q_r/A_r\\
f&=&A_h P_h-A_r P_r.
\eeqa
\end{subequations}
If one deals with rotary actuators, both $A_r$ and $A_h$ (measured in m$^2$) should be replaced by the volume displacement per one radian of rotation, which is measured in m$^3$/rad. 

\par

In conclusion, now we have 18 algebraic constraints in \eq{Q} and 19 variables, which are listed as follows
\btm
\tm $\{Q_{ho},Q_{ro},Q_{hn},Q_{rn},Q_h,Q_r,Q_{ph},Q_{pr},Q_{th},Q_{tr},Q_t,Q_p,Q_b\}$: 13 flowrate values
\tm $\{P,P_h,P_r,P_c\}$: 4 pressure values
\tm $\{f,v\}$: the external force to the rod and the velocity of the rod.
\etm

\begin{remark}
The presented formalism can be said to be close to the classical hydraulic-electric analogy \cite{Greenslade_2003_Analogy,Esposito_1969_Analogy}, which replaces the pressure and the flowrate by the voltage and the current, respectively. Considering that the presented approach involves the nonsmoothness, it is also analogous to the {\it nonsmooth electronics} \cite{Addi_2007_Circuit,Goeleven_2008_Existence,Acary_2010_TimeStepping}. In this analogy, a check valve, for example, corresponds to an ideal diode~\cite{Acary_2010_TimeStepping}. 
\end{remark}

\subsection{Normalized Representations}

For the convenience of derivation, we now normalize some quantities in the following manner:
\beqa
&& q_*\define Q_*/A_h\ (*\in\{ho,hn,h,ph,th\}),\quad \uh_* \define c_*u_*/A_h^{3/2} \ (*\in\{ph,th\})\\
&& q_*\define Q_*/A_r\ (*\in\{ro,rn,r,pr,tr\}),\quad \uh_* \define c_*u_*/A_r^{3/2} \ (*\in\{pr,tr\})\\
&& F_{*}\define P_{*}A_{*}, \quad F_{*M}=P_{*M}A_*\ (*\in\{h,r\})  \\
&& U_b \define c_b u_b.
\eeqa
The regularized input vector is defined as $\bmuh=[\uh_{ph},\uh_{tr},\uh_{pr},\uh_{th},U_b]^T$. By using these definitions, \eq{Q} can be rewritten as follows:
\begin{subequations}
\beqa
&&q_{hn} =q_{ph}-q_{th} \\
&&q_{rn} =q_{pr}-q_{tr} \\
&&q_h =q_{hn}-q_{ho} \\
&&q_r =q_{rn}-q_{ro} \\
&&Q_t =A_hq_{th}+A_r q_{tr} \\
&&Q_p =A_hq_{ph}+A_r q_{pr} \\
&&Q_b =Q-Q_p\dfeq{QboBBB} \\
&&q_{ph} = \uh_{ph}\calR(A_hP_c-F_h) \\
&&q_{th} = \uh_{th}\calR(F_h    ) \\
&&q_{pr} = \uh_{pr}\calR(A_rP_c-F_r) \\
&&q_{tr} = \uh_{tr}\calR(F_r    ) \\
&&Q_b \in U_b\calR(P)+\calN_{(-\infty,P_M]}(P) \dfeq{QboABBB}\\ 
&& Q_p\in\calN_{(-\infty,P_c]}(P) \dfeq{QpBBB} \\
&&q_{ho} \in \calN_{[0,F_{hM}]}(F_h) \\
&&q_{ro} \in \calN_{[0,F_{rM}]}(F_r) \\
&&v=q_h \\
&&v=-q_r \\
&&f=F_h-F_r.\dfeq{ff}
\eeqa
\end{subequations}
Note that the equivalence between \eq{Qp} and \eq{QpBBB} can be derived from the definition of the normal cone. 

\par

Now, let us attempt to reduce the number of variables. Substituting \eq{QpBBB} and \eq{QboBBB} into \eq{QboABBB} yields
\beqa
Q &\in& \calN_{(-\infty,P_c]}(P) + U_b\calR(P)+\calN_{(-\infty,P_M]}(P) \non\\
 &=& U_b\calR(P)+\calN_{(-\infty,\min(P_c,P_M)]}(P),\dfeq{QboBBBccc} 
\eeqa
which is equivalent to
\beqa
P &=&\min(P_c,P_M,\calS(Q/U_b))
\eeqa
because of Theorem~\ref{thm:mainp2}. Since $q_h=-q_r=v$, we can eliminate the two variables $q_h$ and $q_r$. Moreover, $q_{ho}$, $q_{ro}$, $q_{hn}$, $q_{rn}$, $q_{ph}$, $q_{th}$, $q_{pr}$, $q_{tr}$, and $Q_t$ can also be eliminated. This gives the following five equations for the four variables $\{P,P_c,F_h,F_r,f\}$:
\begin{subequations}\dfeq{thesixb}
\beqa
&&-v +\uh_{ph}\calR(A_hP_c-F_h)-\uh_{th}\calR(F_h)\in \calN_{[0,F_{hM}]}(F_h)  \\
&& v +\uh_{pr}\calR(A_rP_c-F_r)-\uh_{tr}\calR(F_r)\in \calN_{[0,F_{rM}]}(F_r)   \\
&&A_h \uh_{ph}\calR(A_hP_c-F_h)+A_r \uh_{pr} \calR(A_rP_c-F_r) \in\calN_{(-\infty,P_c]}(P) \\
&& Q \in U_b\calR(P) + A_h \uh_{ph}\calR(A_hP_c-F_h)+A_r \uh_{pr} \calR(A_rP_c-F_r) +\calN_{(-\infty,P_M]}(P) \\
&& f=F_h-F_r.
\eeqa
\end{subequations}

\subsection{Main Result: Nonsmooth Quasistatic Map from $v$ to $f$}\label{ss:vfrel}

Now we derive the relation between $f$ and $v$ from \eq{thesixb}. If $\bmuh\in\bmcalU_0$, \eq{thesixb} reduces to the following:
\begin{subequations}\dfeq{thesixbzero}
\beqa
-v &\in& \calN_{[0,F_{hM}]}(F_h)  \\
 v &\in& \calN_{[0,F_{rM}]}(F_r)   \\
 0 &\in& \calN_{(-\infty,P_c]}(P) \\
 Q &\in& U_b\calR(P)+\calN_{(-\infty,P_M]}(P)\\
 f &=& F_h-F_r,
\eeqa
\end{subequations}
from which 
\beqa
F_h\in\gsgn(F_{hM},v,0) ,\quad
F_r\in\gsgn(0,v,F_{rM}) \dfeq{fzer}
\eeqa
and $f\in\gsgn(F_{hM},v,-F_{rM})$ can be derived by using \eq{calNgsgn}. This means that $f$ can take any values between $F_{hM}$ and $-F_{rM}$ when $v=0$, which is consistent with the fact that, when all the main control valves are closed, the cylinder holds its position by producing the reaction force against the external force as long as the relief valves are closed.

\par

If $\bmuh\in\bmcalU_+$, \eq{thesixb} reduces to the following:
\begin{subequations}
\beqa
&&-v +\uh_{ph}\calR(A_hP_c-F_h)\in \calN_{[0,F_{hM}]}(F_h)  \dfeq{vq1cc}\\
&& v -\uh_{tr}\calR(F_r)\in \calN_{[0,F_{rM}]}(F_r)  \dfeq{vq2bb} \\
&&A_h \uh_{ph}\calR(A_hP_c-F_h) \in\calN_{(-\infty,P_c]}(P) \dfeq{qp1bb}\\
&& Q \in U_b\calR(P) + A_h \uh_{ph}\calR(A_hP_c-F_h)+\calN_{(-\infty,P_M]}(P) \dfeq{qubs} \\
&& f = F_h-F_r. \dfeq{ffsdf} 
\eeqa
\end{subequations}
Here, \eq{vq1cc}, \eq{vq2bb} and \eq{qp1bb} can be rewritten as
\beqa
&&F_h = \proj_{[0,F_{hM}]} ( A_hP_c-\calS(v)/\uh_{ph}^2 ) \dfeq{prvf} \\
&&F_r = \proj_{[0,F_{rM}]} ( \calS(v)/\uh_{tr}^2 ) \dfeq{prvfvv} \\
&&P_c = \max(P,F_h/A_h), \dfeq{Pcpbma}
\eeqa
respectively, because of Theorem~\ref{thm:mainp2}. Substituting \eq{Pcpbma} into \eq{prvf} and \eq{qubs} results in:
\begin{subequations}\dfeq{thetwo}
\beqa
F_h &=& \proj_{[0,F_{hM}]}(\max(A_hP,F_h)-\calS(v)/\uh_{ph}^2). \dfeq{vq1bbffd} \\
 Q &\in& U_b\calR(P) + A_h \uh_{ph}\calR(\max(A_hP-F_h,0))+\calN_{(-\infty,P_M]}(P), \dfeq{qubsb} 
\eeqa
\end{subequations}
respectively. Now $F_r$ is obtained by \eq{prvfvv}, and thus we focus on obtaining $F_h$ from \eq{thetwo}. If $v<0$, \eq{vq1bbffd} implies $F_h=F_{hM}$. If $v=0$, \eq{vq1bbffd} implies that $\min(A_hP,F_{hM})\leq F_h\leq F_{hM}$. If $v=0$ and $F_h<F_{hM}$, \eq{vq1bbffd} implies $F_h\geq A_hP$ and substituting it into \eq{qubsb} yields $P=\min(P_M,Q^2/U_b^2)$. Therefore, if $v=0$, the condition \eq{thetwo} implies the following:
\beqa
F_h &\in& [\min(F_{hM},A_hP_M,A_hQ^2/U_b^2), F_{hM}]. \dfeq{arf}
\eeqa
If $v>0$, \eq{vq1bbffd} implies $A_hP>F_h$ and thus \eq{thetwo} can be rewritten as follows:
\begin{subequations}\dfeq{thetwob}
\beqa
F_h &=& \proj_{[0,F_{hM}]}(A_hP-\calS(v)/\uh_{ph}^2). \dfeq{fhdf} \\
 Q &\in& U_b\calR(P) + A_h \uh_{ph}\calR(A_hP-\proj_{[0,F_{hM}]}(A_hP-\calS(v)/\uh_{ph}^2))+\calN_{(-\infty,P_M]}(P). \dfeq{sf} \quad
\eeqa
\end{subequations}
The relation \eq{sf} implies that $P$ is a strictly increasing function of $Q$ that saturates at $P=P_M$. As long as $A_hP-\calS(v)/\uh_{ph}^2\in [0,F_{hM}]$, $P$ is written as
\beqa
P=\min\lb(P_M,-\dfrac{A_h^2}{U_b^2}\calS\lb(v-\dfrac{Q}{A_h}\rb)\rb). \dfeq{Pbdf}
\eeqa
Outside this range, $P$ is expressed as another function of $Q$ that strictly increases with respect to $Q$ and continuously connects to \eq{Pbdf}. The value of $P$ outside the range does not influence $F_h$ because $F_h$ is saturated at $0$ and $F_{hM}$ due to \eq{fhdf}. Therefore, if $v>0$, the condition \eq{thetwo} implies 
\beqa
F_h &=& \proj_{[0,F_{hM}]}\lb(\min\lb(A_hP_M,-\dfrac{A_h^3}{U_b^2}\calS\lb(v-\dfrac{Q}{A_h}\rb)\rb)-\dfrac{\calS(v)}{\uh_{ph}^2}\rb), \dfeq{ssdff}
\eeqa
which is obtained by substituting \eq{Pbdf} into \eq{fhdf}. Unifying the three cases, i.e., $F_h=F_{hM}$ for $v<0$, \eq{arf} for $v=0$, and \eq{ssdff} for $v>0$, we have the following expression for the case $\bmuh\in\bmcalU_+$:
\beqa
F_h &\in& \gsgn\lb(
F_{hM}
,v,
\proj_{[0,F_{hM}]}\lb(\min\lb(A_hP_M,-\dfrac{A_h^3}{U_b^2}\calS\lb(v-\dfrac{Q}{A_h}\rb)\rb)-\dfrac{\calS(v)}{\uh_{ph}^2}\rb)\rb). \dfeq{Fhsdff}
\eeqa

\par

If $\bmuh\in\bmcalU_-$, we have 
\beqa
F_h= \proj_{[0,F_{hM}]} ( -\calS(v)/\uh_{th}^2 ) \dfeq{Fhdff}
\eeqa
for any $v$ in the same manner as \eq{prvfvv}. Combining \eq{Fhdff} for $\bmuh\in\bmcalU_-$, \eq{Fhsdff} for $\bmuh\in\bmcalU_+$ and \eq{fzer} for $\bmuh\in\bmcalU_0$, we have
\begin{subequations}\dfeq{defGammah}
\beqa
F_h &\in& \Gamma_h(v) \define \gsgn\lb(\Gamma_{h-}(v) ,v,\Gamma_{h+}(v) \rb)
\eeqa
where
\beqa
\Gamma_{h+}(v) &\define& \proj_{[0,F_{hM}]}\lb(\min\lb(A_hP_M,-\dfrac{A_h^3}{U_b^2}\calS\lb(v-\dfrac{Q}{A_h}\rb)\rb)-\dfrac{\calS(v)}{\uh_{ph}^2}\rb)\\
\Gamma_{h-}(v) &\define& \proj_{[0,F_{hM}]}\lb( -\dfrac{\calS(v)}{\uh_{th}^2}\rb) 
\eeqa
\end{subequations}
for any $\bmuh$. In the same manner, $F_r$ for any $\bmuh$ is obtained as follows:
\begin{subequations}\dfeq{defGammar}
\beqa
F_r &\in&\Gamma_r(v)  \define \gsgn\lb(\Gamma_{r-}(v),v,\Gamma_{r+}(v)\rb)
\eeqa
where
\beqa
\Gamma_{r+}(v) &\define&
\proj_{[0,F_{rM}]}\lb( \dfrac{\calS(v)}{\uh_{tr}^2}\rb)\\
\Gamma_{r-}(v) &\define&
\proj_{[0,F_{rM}]}\lb(\min\lb(A_rP_M,\dfrac{A_r^3}{U_b^2}\calS\lb(v+\dfrac{Q}{A_r}\rb)\rb)+\dfrac{\calS(v)}{\uh_{pr}^2}\rb).
\eeqa
\end{subequations}
By using \eq{defGammah} and \eq{defGammar}, the relation between $v$ and $f$ is obtained as follows:
\beqa
f&\in&\Gamma(v) \define \Gamma_h(v)-\Gamma_r(v).
\eeqa

\par

For the convenience of further derivations, we can also write the set-valued map $\Gamma(v)$ in the following form:
\begin{subequations}\dfeq{Gammapp}
\beqa
\Gamma(v)&=&\gsgn(\Gamma_-(v),v,\Gamma_+(v)) \dfeq{Gammagsgn}
\eeqa
where
\beqa
\Gamma_+(v)&\define& \Gamma_{h+}(v)-\Gamma_{r+}(v)
\\
&=& \max(\min( \max(\Gamma_{+0a}(v),\Gamma_{+0b}(v)),\max(\Gamma_{+1a}(v),\Gamma_{+1b}(v)),\non\\
&&\quad \max(\Gamma_{+2a}(v),\Gamma_{+2b}(v))),\Gamma_{+3}(v),-F_{rM}) \dfeq{defgammapos} \quad\\
\Gamma_-(v)&\define& \Gamma_{h-}(v) -\Gamma_{r-}(v) \\
&=& \min(\max(\min(\Gamma_{-0a}(v),\Gamma_{-0b}(v)),\min(\Gamma_{-1a}(v),\Gamma_{-1b}(v)),\non\\
&&\quad \min(\Gamma_{-2a}(v),\Gamma_{-2b}(v))),\Gamma_{-3}(v), F_{hM}) \quad\\
\Gamma_{+0a}(v)&\define&F_{hM}- \dfrac{\calS(v)}{\uh_{tr}^2} \dfeq{Gammappeach0}\\
\Gamma_{+0b}(v)&\define&F_{hM}-F_{rM}\\
\Gamma_{+1a}(v)&\define&-\dfrac{A_h^3}{U_b^2}\calS\lb(v-\dfrac{Q}{A_h}\rb)-\dfrac{\calS(v)}{\uh_{ph}^2}- \dfrac{\calS(v)}{\uh_{tr}^2}\\
\Gamma_{+1b}(v)&\define&-\dfrac{A_h^3}{U_b^2}\calS\lb(v-\dfrac{Q}{A_h}\rb)-\dfrac{\calS(v)}{\uh_{ph}^2}-F_{rM}\\
\Gamma_{+2a}(v)&\define&A_hP_M-\dfrac{\calS(v)}{\uh_{ph}^2} -\dfrac{\calS(v)}{\uh_{tr}^2} \\
\Gamma_{+2b}(v)&\define&A_hP_M-\dfrac{\calS(v)}{\uh_{ph}^2} -F_{rM} \\
\Gamma_{+3}(v)&\define&-\dfrac{\calS(v)}{\uh_{tr}^2}\dfeq{Gammappeach1} \\
\Gamma_{-0a}(v)&\define&-F_{rM}-\dfrac{\calS(v)}{\uh_{th}^2}\\
\Gamma_{-0b}(v)&\define&-F_{rM}+F_{hM} \\
\Gamma_{-1a}(v)&\define&-\dfrac{A_r^3}{U_b^2}\calS\lb(v+\dfrac{Q}{A_r}\rb)-\dfrac{\calS(v)}{\uh_{pr}^2}-\dfrac{\calS(v)}{\uh_{th}^2}\\
\Gamma_{-1b}(v)&\define&-\dfrac{A_r^3}{U_b^2}\calS\lb(v+\dfrac{Q}{A_r}\rb)-\dfrac{\calS(v)}{\uh_{pr}^2}+F_{hM} \\
\Gamma_{-2a}(v)&\define&-A_rP_M-\dfrac{\calS(v)}{\uh_{pr}^2}-\dfrac{\calS(v)}{\uh_{th}^2} 
\\
\Gamma_{-2b}(v)&\define&-A_rP_M-\dfrac{\calS(v)}{\uh_{pr}^2}+F_{hM}
\\
\Gamma_{-3}(v)&\define&-\dfrac{\calS(v)}{\uh_{th}^2}. 
\eeqa
\end{subequations}
Here, the divisions by $\uh_{\ast}^2$ or $U_b^2$ do not cause troubles even when they are zeros because \eq{defgammapos} implicitly includes upper and lower bounds. One workaround in the implementation may be to replace the divisions by $\uh_{\ast}^2$ by those by $\max(\varepsilon,\uh_{\ast}^2)$ where $\varepsilon$ is a small positive number close to the machine epsilon.

\par

\begin{table}[t!]\centering
\caption{States of valves at each segment of $\Gamma(v)$. hR = head-side relief valve; hSC = head-side suction check valve; pR = pump relief valve; rSC = rod-side suction relief valve; rR = rod-side relief valve; o = open; x = closed; - = either open or closed.}\label{tb:valvestates}\vspace{5pt}
\begin{tabular}{cccccccccc}\hline
$\Gamma$               &$v$    &$\bmuh$                              &hR &hSC&pR &rSC&rR \\\hline
$-F_{rM}$              &$+$    &$\bmcalU_+\cup\bmcalU_0\cup\bmcalU_-$& x & o & - & x & o \\\hline
$\Gamma_{+3} $         &$+$    &$\bmcalU_+$                          & x & o & - & x & x \\\hline
$\Gamma_{+2b}$         &$+$    &$\bmcalU_+$                          & x & x & o & x & o \\\hline
$\Gamma_{+2a}$         &$+$    &$\bmcalU_+$                          & x & x & o & x & x \\\hline
$\Gamma_{+1b}$         &$+$    &$\bmcalU_+$                          & x & x & x & x & o \\\hline
$\Gamma_{+1a}$         &$+$    &$\bmcalU_+$                          & x & x & x & x & x \\\hline
$\Gamma_{+0b}$         &$+$    &$\bmcalU_+$                          & o & x & - & x & o \\\hline
$\Gamma_{+0a}$         &$+$    &$\bmcalU_+$                          & o & x & - & x & x \\\hline
$[\Gamma_+(0),F_{hM}]$ &$0$    &$\bmcalU_+$                          & x & x & - & x & x \\\hline
$[-F_{rM},F_{hM}]$     &$0$    &$\bmcalU_0$                          & x & x & - & x & x \\\hline
$[-F_{rM},\Gamma_-(0)]$&$0$    &$\bmcalU_-$                          & x & x & - & x & x \\\hline
$\Gamma_{-0a}$         &$-$    &$\bmcalU_-$                          & x & x & - & x & o \\\hline
$\Gamma_{-0b}$         &$-$    &$\bmcalU_-$                          & o & x & - & x & o \\\hline
$\Gamma_{-1a}$         &$-$    &$\bmcalU_-$                          & x & x & x & x & x \\\hline
$\Gamma_{-1b}$         &$-$    &$\bmcalU_-$                          & o & x & x & x & x \\\hline
$\Gamma_{-2a}$         &$-$    &$\bmcalU_-$                          & x & x & o & x & x \\\hline
$\Gamma_{-2b}$         &$-$    &$\bmcalU_-$                          & o & x & o & x & x \\\hline
$\Gamma_{-3} $         &$-$    &$\bmcalU_-$                          & x & x & - & o & x \\\hline
$F_{hM}$               &$-$    &$\bmcalU_+\cup\bmcalU_0\cup\bmcalU_-$& o & x & - & o & x \\\hline
\end{tabular}
\end{table}

\par

It should be noted that, at $v=0$, the function $\Gamma(v)$ is set-valued and its value is the closed set $[\Gamma_+(0),\Gamma_-(0)]$. The boundaries $\Gamma_+(0)$ and $\Gamma_-(0)$ are obtained from a straightforward derivation as
\begin{subequations}\dfeq{gammazero}
\beqa
\Gamma_+(0) = \Gamma_{h+}(0) -\Gamma_{r+}(0) ,\quad 
\Gamma_-(0) = \Gamma_{h-}(0) -\Gamma_{r-}(0) 
\eeqa
where
\beqa
\Gamma_{h+}(0) &=&\lb\{\barr{ll}
\min(F_{hM},A_hP_M,A_hQ^2/U_b^2) &\mbox{if } \uh_{ph}>0 \\
0                                &\mbox{if } \uh_{ph}=0 \\
\earr\rb.
\\
\Gamma_{r+}(0) &=&\lb\{\barr{ll}
0      &\mbox{if } \uh_{tr}>0 \\
F_{rM} &\mbox{if } \uh_{tr}=0 \\
\earr\rb. \dfeq{gammahp0}
\\
\Gamma_{h-}(0) &=&\lb\{\barr{ll}
0      &\mbox{if } \uh_{th}>0 \\
F_{hM} &\mbox{if } \uh_{th}=0 \\
\earr\rb. 
\\
\Gamma_{r-}(0) &=&\lb\{\barr{ll}
\min(F_{rM},A_rP_M,A_rQ^2/U_b^2) &\mbox{if } \uh_{pr}>0 \\
0                                &\mbox{if } \uh_{pr}=0. \\
\earr\rb. 
\eeqa
\end{subequations}

As can be seen from the expression \eq{Gammapp}, the function $\Gamma(v)$ is composed of many segments. Interestingly, each segment has a clear physical interpretation. For example, in the segments listed in \eq{Gammappeach0}-\eq{Gammappeach1}, the last term, either $-F_{rM}$ or $-\calS(v)/\uh_{tr}^2$, represents the state of the rod-side relief valve, either closed or open, respectively. In this way of consideration, we can summarize the valve states at each curve segment as in Table~\ref{tb:valvestates}.

\par

\subsection{Numerical Examples}\label{ss:numestt}

\begin{figure}[t!]\begin{center}
\includegraphics[scale=0.8]{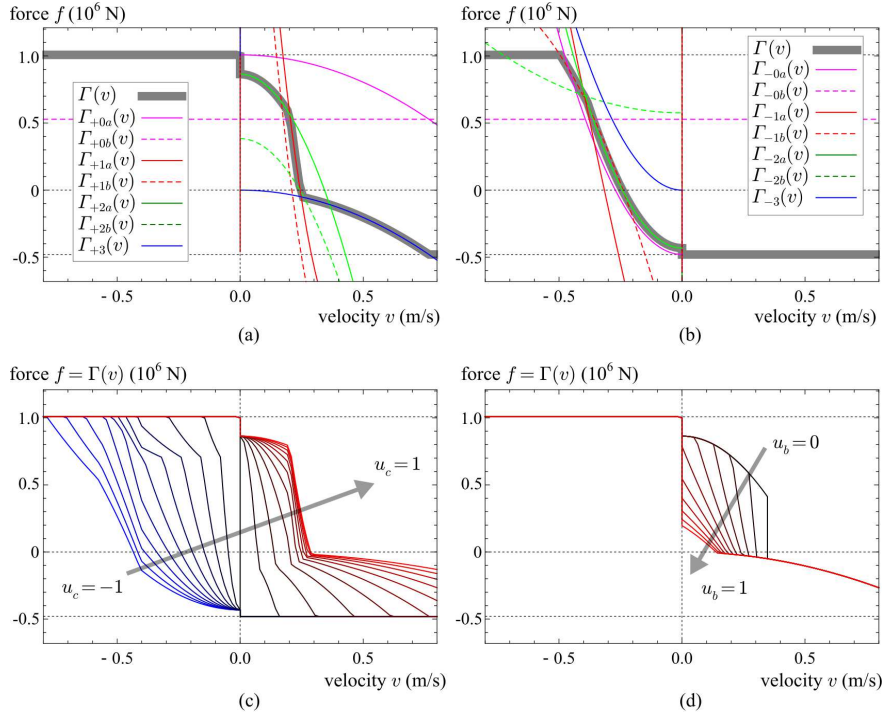}\\[-4pt]
\caption{Numerical examples of the quasistatic map $f\in\Gamma(v)$. (a) $u_c=0.5$, $u_b=0.2$; (b) $u_c=-0.5$, $u_b=0.2$; (c) $u_c\in[-1,1]$, $u_b=0.2$; (d) $u_c=0.7$, $u_b\in[0,1]$. The velocity $v$ is positive when the rod is extending. The force $f$ is positive when the external force is compressive and the actuator force acts to extend the rod. }\label{fg:nume_stt}
\end{center}\end{figure}

Some numerical examples are now presented. We consider an asymmetric hydraulic cylinder with the following parameter values: 
\beqa
&& C_* = 0.6,\ a_* = 0.0001~\mbox{m$^2$}, \ \rho = 850~\mbox{kg/m$^3$}, \non\\
&& A_r = 0.012~\mbox{m$^2$},\  A_h = 0.024~\mbox{m$^2$}, \non\\ 
&& P_{rM} = 40~\mbox{MPa},\  P_{hM} = 42~\mbox{MPa}, \ P_M = 36~\mbox{MPa}, \non\\
&& Q = 500~\mbox{L/min} = 0.00833~\mbox{m$^3$/s}. 
\eeqa
A common control command $u_c\in[-1,1]$ is used to set the openings of the main control valves as follows:
\beqa
&& u_{ph} = u_{tr} = \max(u_c,0),\ 
 u_{pr} = u_{th} = \max(-u_c,0). \dfeq{ucgiven}
\eeqa
The opening of the bleed valve was fixed at $u_b=0.2$ unless otherwise noted.

\par

Results with different values of $u_c$ and $u_b$ are shown in Fig.~\ref{fg:nume_stt}. Fig.~\ref{fg:nume_stt}(a) shows the function $\Gamma$ and its segments with a positive $u_c$, while Fig.~\ref{fg:nume_stt}(b) shows those with a negative $u_c$. They show that $\Gamma(v)$ is always a decreasing function of $v$ and it is set-valued at $v=0$. Fig.~\ref{fg:nume_stt}(c) shows how the function $\Gamma(v)$ varies according to the change in $u_c$. It shows that, at a constant external force, the velocity $v$ increases as $u_c$ increases, which is consistent with the behavior of real hydraulic actuator. Fig.~\ref{fg:nume_stt}(d) shows that the velocity $v$ decreases as the bleed valve opening $u_b$ increases. It can be explained by the fact that the oil flow into the actuator decreases as the bleed valve opens. 

\par

\section{Incorporation into Multibody Simulators}\label{sc:dyn}

\subsection{Virtual Viscoelastic Element}\label{ss:relax}

The previous section has modeled a hydraulic actuator as a set-valued function $\Gamma$ from the velocity $v$ to the force $f$. Because of the set-valuedness, the function $\Gamma$ is not convenient for the use in simulations. The set-valuedness at $v=0$ is an important feature that cannot be neglected because the actuator actually stops when the oil flow is blocked by the valves. This section shows an approach to incorporate the set-valued function $\Gamma$ in multibody simulators.

\par

\begin{figure}[t!]\begin{center}
\includegraphics[scale=1.2]{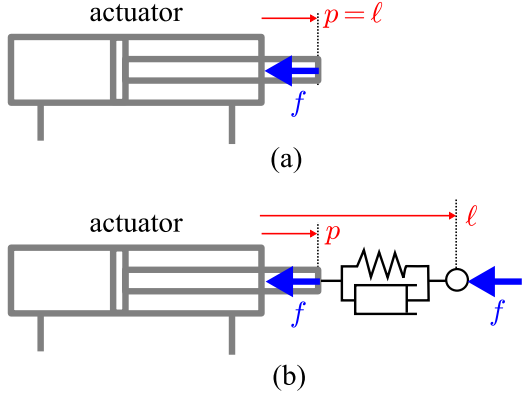}\\[-4pt]
\caption{Differential-algebraic relaxation using a virtual viscoelastic element. (a) The original system with the algebraic constraint $f\in\Gamma(\dot{\ell})$. (b) The approximate (relaxed) system with the differential-algebraic constraint \eq{thedae}.}
\label{fg:relax}
\end{center}\end{figure}

The quasistatic model presented in the previous section describes the system of Fig.~\ref{fg:relax}(a) as an algebraic constraint $f\in\Gamma(\dot{p})$ where $p$ is the displacement of the actuator rod and $v=\dot{p}$. To deal with the set-valuedness, we consider installing a virtual viscoelastic element connected with the actuator as illustrated in Fig.~\ref{fg:relax}(b). Let $\ell$ be the total displacement including the viscoelastic element and $p$ be the displacement of the actuator rod. The stretch of the viscoelastic element is $\ell-p$. One can see that, if the viscoelastic element is stiff enough, $p\approx \ell$ is maintained and thus the system of Fig.~\ref{fg:relax}(b) approximates the system of Fig.~\ref{fg:relax}(a). The dynamics of the system of Fig.~\ref{fg:relax}(b) can be described as follows:
\begin{subequations}\dfeq{thedae}
\beqa
f &=& K(p-\ell)+B(\dot{p}-\dot{\ell}) \dfeq{KpBdp} \\
f&\in& \Gamma(\dot{p}). \dfeq{fgmmv}
\eeqa
\end{subequations}
Here, $K$ and $B$ are the stiffness and the viscosity of the virtual viscoelastic element. Because there is no inertia between the actuator and the viscoelastic element, $f$ in \eq{KpBdp}, determined by the viscoelastic element, is always the same as $f$ in \eq{fgmmv}, acting on the actuator. Although some discussions on the choice of $K$ and $B$ will be given in Remark~\ref{rm:choicekb}, a simple guideline for a better approximation (i.e., $p\approx \ell$) is that the stiffness $K$ should be as high as the simulation is stable, and that the viscosity $B$ should be high enough to damp the oscillation in the relative displacement $p-\ell$. The expression \eq{thedae} can be seen as a differential-algebraic inclusion (DAI) with respect to $p$. The quantities $\ell$ and $\dot{\ell}$ are given as inputs, and $\dot{p}$ and $f$ are determined by the pair \eq{thedae} of algebraic constraints.

\par

To provide the solution of these algebraic constraints \eq{thedae}, let us consider a function $\Lambda$ that satisfies the following relation:
\beqa
 \beta v + \fb \in\Gamma(v) &\iff&  v = \Lambda(\beta,\fb) \dfeq{deflambda}
\eeqa
where $\beta>0$. Because $\Gamma(v)$ is a decreasing function of $v$ and $\beta v+\fb$ is a strictly increasing function of $v$, the algebraic inclusion $\beta v+\fb \in\Gamma(v)$ always has a unique solution with respect to $v$, and thus $\Lambda(\beta,\fb)$ is a globally single-valued function. Fortunately, $\Lambda(\beta,\fb)$ can be obtained analytically as will be shown in the next Section~\ref{ss:Lambda}. Considering that \eq{KpBdp} is equivalent to $f=B\dot{p}+\fb$ where $\fb=K(p-\ell)-B\dot{\ell}$, the DAI \eq{thedae} can be equivalently rewritten as follows:
\begin{subequations}\dfeq{odeeq}
\beqa
\dot{p} &=& \Lambda(B, K(p-\ell)-B\dot{\ell}) \\
f &=& K(p-\ell)+B(\dot{p}-\dot{\ell}).
\eeqa
\end{subequations}
The expression \eq{odeeq} is algebraically equivalent to the expression \eq{thedae}, and can be seen as a closed-form solution of the algebraic problem \eq{thedae}. Moreover, because $\Lambda$ is a continuous, single-valued function, \eq{odeeq} is only an ordinary differential equation (ODE) representing a first-order system of which the state is $p$. That is, \eq{odeeq} can be used for simulations combined with common numerical integration schemes.

\par

Another numerical scheme for the DAI \eq{thedae} can be obtained through the discretization along time. With the implicit (backward) Euler discretization, i.e., $v=\dot{p}\approx(p_k-p_{k-1})/h$ and $\dot{\ell}\approx(\ell_k-\ell_{k-1})/h$, one can discretize \eq{thedae} as follows:
\begin{subequations}\dfeq{thedaedisc}
\beqa
f_k &=& (B+hK)v+K(p_{k-1}-\ell_k)-B(\ell_k-\ell_{k-1})/h \dfeq{KpBdpb} \\
f_k &\in& \Gamma(v) \dfeq{fgamma}\\
p_k &=&  p_{k-1} + hv 
\eeqa
\end{subequations}
where $k$ denote the discrete-time index and $h$ denotes the timestep size. The pair of constraints \eq{KpBdpb} and \eq{fgamma} poses an algebraic problem with respect to $v$ and $f_k$, and the function $\Lambda$ provides the solutions. That is, we can solve \eq{thedaedisc} in the following algorithm:
\begin{subequations}\dfeq{algo}
\beqa
\bar{f} &:=& K(p_{k-1}-\ell_k)- B(\ell_k - \ell_{k-1})/h  \\
v &:=& \Lambda(B+hK,\fb) \\
p_k &:=& p_{k-1}+ h v \\
f_k &:=& \fb  + (B+hK)v.
\eeqa
\end{subequations}
This algorithm accepts the inputs $\{\ell_k,p_{k-1},\ell_{k-1}\}$ and provides the output $f_k$, and can be seen as a numerical integration scheme for the DAI \eq{thedae}. In the simulation, the values of $\{p_k,\ell_k\}$ should be stored and used as $\{p_{k-1},\ell_{k-1}\}$ in the next timestep.

\par

In summary, the original set-valued constraint $f\in\Gamma(\dot{\ell})$ is approximated by the DAI \eq{thedae}, which is implementable as the ODE \eq{odeeq} or the discrete-time algorithm \eq{algo}. For the implementation, the continuous-time form \eq{odeeq} would be preferred for rather complicated simulations with sophisticated ODE solvers, while the Euler-based discrete-time form \eq{algo} would be convenient for rather simple simulations that do not require high accuracy. This approximation scheme, which relaxes a set-valued algebraic constraint by a DAI, is what Kikuuwe~\cite{Kikuuwe_2018_VSS} has called a {\it differential algebraic relaxation}. The scheme has been utilized to deal with Coulomb friction in simulators \cite{Kikuuwe_2006_TRO,Xiong_2013_Tribo,Xiong_2013_JAMH,Kikuuwe_2019_MPE} and to implement sliding mode controllers to discrete-time systems \cite{Kikuuwe_2010_TRO,Kikuuwe_2014_TMECH}.

\begin{remark}\label{rm:choicekb}
The choice of $K$ and $B$ can be discussed from two different points of view. On the one hand, if one needs to approximate the ideal situation where the actuator is rigidly connected to external components (such as links of an excavator), the stiffness $K$ should be set as high as the simulation is stable. On the other hand, if one needs to simulate a particular real system that comprises the compliance of the components and the compressibility of the oil, the values of $K$ and $B$ should be chosen so that the simulator's response is close to that of the real system. In this case, the values of $K$ and $B$ may be considered as the aggregation of the compliance and damping of the components and the oil. The tuning procedure might be performed by explicitly considering these factors, or on a trial-and-error basis, comparing the data from the simulator and the real system.
\end{remark}

\subsection{Analytical Form of Function $\Lambda$}\label{ss:Lambda}

Now we present an analytical form of $\Lambda(\beta,\fb)$, of which the definition has been given only implicitly as in \eq{deflambda}. Let $v_{\ast}$ denote the solution of $\beta v+\fb=\Gamma_{\ast}(v)$ with respect to $v$ where $\Gamma_{\ast}(v)$ are those listed in \eq{Gammapp}. By observing the definitions of $\Gamma_{\ast}(v)$ and using Proposition~\ref{prp:PhiA} and Proposition~\ref{prp:PhiB}, we can obtain the followings:
\begin{subequations}\dfeq{defvvvv}
\beqa
v_{+0a}&\define& \Phi_A\lb(\beta,\fb-F_{hM},\uh_{tr}^2\rb)\\ 
v_{+0b}&\define& (F_{hM}-F_{rM}-\fb)/\beta\\       
v_{+1a}&\define& \Phi_B\lb(\beta ,\fb, \psi(\uh_{ph},\uh_{tr}),\dfrac{U_b^2}{A_h^3}, \dfrac{Q}{A_h}\rb)\\ 
v_{+1b}&\define& \Phi_B\lb(\beta ,\fb+F_{rM},\uh_{ph}^2,\dfrac{U_b^2}{A_h^3}, \dfrac{Q}{A_h}\rb)\\       
v_{+2a}&\define& \Phi_A\lb(\beta ,\fb-A_hP_M, \psi(\uh_{ph},\uh_{tr})\rb)\\             
v_{+2b}&\define& \Phi_A\lb(\beta ,\fb-A_hP_M+F_{rM}, \uh_{ph}^2\rb)\\             
v_{+3} &\define& \Phi_A\lb(\beta ,\fb, \uh_{tr}^2 \rb) \\
v_{rM} &\define& (-F_{rM}-\fb)/\beta \\
v_{-0a}&\define& \Phi_A\lb(\beta,\fb+F_{rM},\uh_{th}^2\rb)\\ 
v_{-0b}&\define& (F_{hM}-F_{rM}-\fb)/\beta\\            
v_{-1a}&\define& \Phi_B\lb(\beta ,\fb,\psi(\uh_{pr},\uh_{th}),\dfrac{U_b^2}{A_r^3}, -\dfrac{Q}{A_r}\rb)\\          
v_{-1b}&\define& \Phi_B\lb(\beta ,\fb-F_{hM}       ,\uh_{pr}^2,\dfrac{U_b^2}{A_r^3}, -\dfrac{Q}{A_r}\rb)\\          
v_{-2a}&\define& \Phi_A\lb(\beta ,\fb+A_rP_M       ,\psi(\uh_{pr},\uh_{th})\rb)\\             
v_{-2b}&\define& \Phi_A\lb(\beta ,\fb+A_rP_M-F_{hM},\uh_{pr}^2\rb)\\             
v_{-3} &\define& \Phi_A\lb(\beta ,\fb              ,\uh_{th}^2 \rb) \\
v_{hM}    &\define& (F_{hM}-\fb)/\beta.
\eeqa
\end{subequations}
Then, by using Theorem~\ref{thm:fgswitch} and Theorem~\ref{thm:main}, the function $\Lambda$ can be obtained as follows:
\begin{subequations}\dfeq{defLambdadetail}
\beqa
\Lambda(\beta,\fb) &\define& \lb\{\barr{ll}
 v_{hM}              &\mbox{if } \fb > F_{hM}     \,\wedge\,\bmuh\in(\bmcalU_+\cup\,\bmcalU_0) \\
 \Lambda_+(\beta,\fb)&\mbox{if } \fb < \Gamma_+(0)\,\wedge\,\bmuh\in\bmcalU_+ \\
 v_{rM}              &\mbox{if } \fb < -F_{rM}    \,\wedge\,\bmuh\in(\bmcalU_-\cup\,\bmcalU_0) \\
 \Lambda_-(\beta,\fb)&\mbox{if } \fb > \Gamma_-(0)\,\wedge\,\bmuh\in\bmcalU_- \\
 0                   &\mbox{if } \Gamma_+(0)\leq\fb\leq\Gamma_-(0)
\earr\rb.
\eeqa
where 
\beqa
\Lambda_+(\beta,\fb)&\define& \max(\min(\max(v_{+0a},v_{+0b}),\max(v_{+1a},v_{+1b}), 
\max(v_{+2a},v_{+2b})),v_{+3},v_{rM}) \\
\Lambda_-(\beta,\fb)&\define&  \min(\max(\min(v_{-0a},v_{-0b}),\min(v_{-1a},v_{-1b}),
\min(v_{-2a},v_{-2b})),v_{-3},v_{hM}).
\eeqa
\end{subequations}
Here, $\Gamma_+(0)$ and $\Gamma_-(0)$ are those given by \eq{gammazero}.

\par

\subsection{Numerical Examples}\label{ss:numedyn}

\begin{figure}[t!]\begin{center}
\includegraphics[scale=1.0]{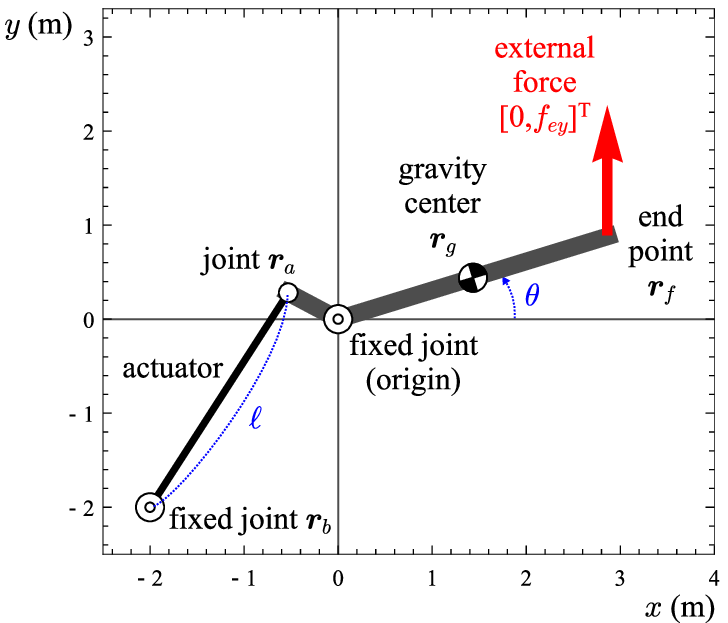}\\[-4pt]
\caption{A single-DOF arm driven by a hydraulic actuator.}\label{fg:dynamic_sys}
\vspace{20pt}
\includegraphics[scale=0.8]{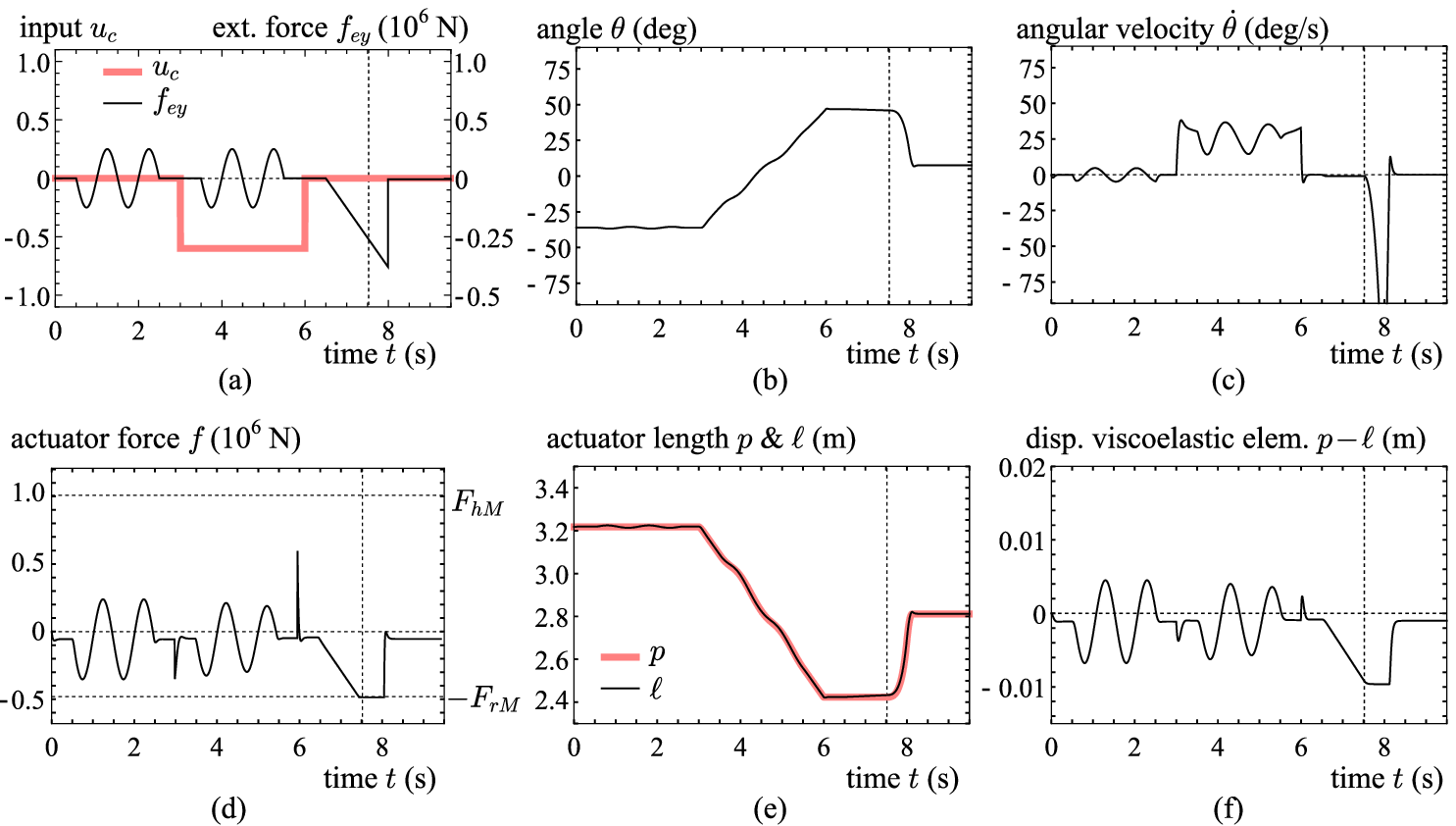}\\[-4pt]
\caption{Inputs (a) and results (b)-(f) of the simulation of the system of Fig.~\ref{fg:dynamic_sys}. The actuator force $f$ is positive when it acts to extend the rod. The actuator force $f$ saturates at $t=7.404$~s, indicated by the vertical lines.}\label{fg:nume_dyn}
\end{center}\end{figure}

The method is illustrated with a simple dynamics simulation of a single degree-of-freedom (DOF) arm system shown in Fig.~\ref{fg:nume_dyn}. The equation of motion is written as follows:
\beqa
(J+M L_g^2)\ddot{\theta} &=&    \bmr_g\times \bmtx{c} 0 \\ - M g \emtx + \bmr_m\times \lb(\dfrac{\bmr_m-\bmr_b}{\|\bmr_m-\bmr_b\|} f\rb) +\bmr_f\times \bmtx{c} 0 \\ f_{ey} \emtx \dfeq{mecdyna} \\
  \ell &=& \| \bmr_m-\bmr_b \| \dfeq{elcalc}\\
  f &=& K(\ell - p ) + B(\dot{\ell} - \dot{p})  \\  
  f &\in& \Gamma(\dot{p}) 
\eeqa
where $\bmr_*$ are those indicated in Fig.~\ref{fg:dynamic_sys}, which are calculaed as follows:  
\beqa
\bmr_g = \bmtx{c} L_g\cos\theta \\ L_g\sin\theta \emtx, \quad
\bmr_m = \bmtx{c} -L_m\cos(\theta-\alpha) \\ -L_m\sin(\theta-\alpha)\emtx,\quad \dfeq{rgrm}
\bmr_f = \bmtx{c} L_f\cos\theta\\ L_f\sin\theta\emtx.
\eeqa
Here, $\times$ stands for a pseudo cross-product of two-dimensional vectors, $[a_1,a_2]^T\times[b_1,b_2]^T=a_1b_2-a_2b_1$. The parameters were set as $L_g=1.5$~m, $L_m=0.6$~m, $\alpha=\pi/4$, $L_f=3$~m, $M=2000$~kg, and $J=5000$~kg$\cdot$m$^2$. The parameters for the virtual viscoelastic elements were chosen as $K=5\times 10^7$~N/m and $B=2.5\times 10^6$~N/m, which as high as the simulation was numerically stable. The parameters of the actuator were set the same as in Section~\ref{ss:numestt}. The openings of the main control valves were given by \eq{ucgiven} with the command $u_c\in[-1,1]$, and the bleed valve opening was set as $u_b=0.3$. The timestep size was set as $h=0.001$~s. At every timestep, the following procedure was performed:
\ben
\tm Calculate $\bmr_g$, $\bmr_m$, and $\bmr_f$ according to \eq{rgrm} with $\theta$ determined by the previous timestep, and calculate $\ell$ according to \eq{elcalc}.
\tm Calculate $f$ and update $p$ according to the algorithm \eq{algo} with the command $u_c$ and the obtained $\ell$. 
\tm Update $\theta$ and $\dot{\theta}$ according to the equation of motion \eq{mecdyna} with the inputs $f$ and $f_{ey}$.
\een
The initial value of $p$ was set equal to that of $\ell$.

\par

In the simulation, the command $u_c$ and $f_{ey}$ were given as indicated in Fig.~\ref{fg:nume_dyn}(a). The results are shown in Figs.~\ref{fg:nume_dyn}(b)-(f). From $t=0$~s to $3$~s, the valves are closed and thus the arm holds its angle even under the fluctuating external force $f_{ey}$. From $t=3$~s to $6$~s, the arm moves up due the negative $u_c$ (to retract the rod). Fig.~\ref{fg:nume_dyn}(c) shows that the velocity is more fluctuating in this period than when $u_c=0$. This behavior is consistent with the fact that hydraulic actuators become more compliant to external forces when the valves are open and the oil is permitted to flow. After $t=6$~s, the valves are closed again, and the increasing external force does not cause movement of the arm until $t=7.404$~s, with the increasing actuator force $f$ balancing the external force. At $t=7.404$~s, the actuator force $f$ reaches the value of $-F_{rM}$ at which the rod-side relief valve opens, and the arm starts to move down. In the whole period, the displacement $p-\ell$ of the viscoelastic element is maintained small (less than a few centimeters), exhibiting the validity of the relaxation scheme introduced in Section~\ref{ss:relax}. In conclusion, the results support the usefulness of the proposed model and the relaxation scheme for reproducing qualitative features of hydraulic actuators.

\begin{remark}\label{rm:choicekbnume}
The transient responses seen in Fig.~\ref{fg:nume_dyn}, such as the spikes in the actuator force $f$ at $t\approx 3$~s and $t\approx 6$~s, are determined by the values of $K$ and $B$. As discussed in Remark~\ref{rm:choicekb}, one needs to tune the values of $K$ and $B$ so that the simulator exhibits similar responses to the real system.
\end{remark}

\section{Extension 1: Regeneration Circuit}\label{sc:reg}
\subsection{Quasistatic Model}\label{ss:regstt}

\begin{figure}[t!]\begin{center}
\includegraphics[scale=1.2]{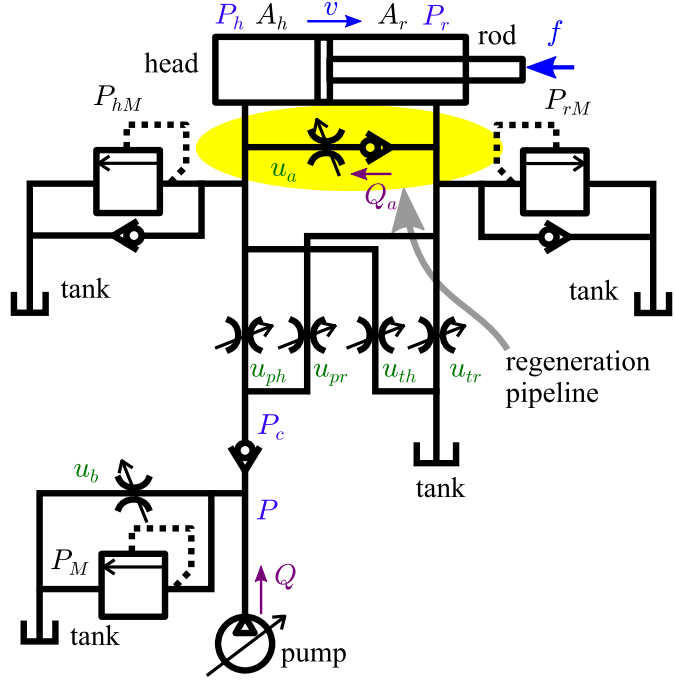}\\[-4pt]
\caption{Regeneration pipeline added to the circuit of Fig.~\ref{fg:circuit}.}\label{fg:reg_circuit}
\end{center}\end{figure}

In some practical applications, the circuit of Fig.~\ref{fg:circuit} includes an additional pathway from the rod-side chamber to the head-side chamber to make the extending movement faster. This section considers the circuit of Fig.~\ref{fg:reg_circuit}, which is an extension of the circuit of Fig.~\ref{fg:circuit}. We assume that the regeneration pipeline has a valve that flows only from the head side to the rod side, and assume that the cross-sectional areas of the cylinder satisfy
\beqa
A_h \geq A_r. \dfeq{assump}
\eeqa
There must be some cylinders that do not satisfy \eq{assump}, but the removal of this assumption is left for future study. Let $a_a$ is the maximum opening area (m$^2$) of the regeneration valve, and let $u_a\in[0,1]$ be the dimensionless input value, which is the ratio of the valve opening area to its maximum value $a_a$. Then, the flowrate of the oil through the regeneration pipeline is written as follows:
\beqa
Q_a &=& c_au_a\max(\calR(F_r/A_r-F_h/A_h,0)) \dfeq{qaua}
\eeqa
where $c_a \define C_a a_a \sqrt{2/\rho}$ where $C_a$ is the discharge coefficient of the regeneration valve, typically around $0.6$ or $0.7$.

\par

With this regeneration pipeline and its flowrate $Q_a$, \eq{thesixb} is extended into the following:%
\begin{subequations}\dfeq{thesixbreg}
\beqa
&&-v +\uh_{ph}\calR(A_hP_c-F_h)-\uh_{th}\calR(F_h)\in \calN_{[0,F_{hM}]}(F_h) - Q_a/A_h   \\
&& v +\uh_{pr}\calR(A_rP_c-F_r)-\uh_{tr}\calR(F_r)\in \calN_{[0,F_{rM}]}(F_r) + Q_a/A_r  \\
&&A_h \uh_{ph}\calR(A_hP_c-F_h)+A_r \uh_{pr} \calR(A_rP_c-F_r) \in\calN_{(-\infty,P_c]}(P) \\
&& Q \in U_b\calR(P) + A_h \uh_{ph}\calR(A_hP_c-F_h)+A_r \uh_{pr} \calR(A_rP_c-F_r) +\calN_{(-\infty,P_M]}(P) \\
&& f=F_h-F_r.
\eeqa
\end{subequations}
With $Q_a=0$, which results from $u_a=0$, \eq{thesixbreg} reduces to \eq{thesixb}. Let us assume that $u_a$ is set positive only when $\bmuh\in\bmcalU_+$ because it is the case in many practical hydraulic circuits. In addition, we assume that $u_{tr}>0$ for the simplicity. Under these conditions, \eq{thesixbreg} and \eq{qaua} reduce to the following:
\begin{subequations}\dfeq{thesixbregp}
\beqa
&& \hat{A} v_a-v +\uh_{ph}\calR(A_hP_c-F_h)\in \calN_{[0,F_{hM}]}(F_h) \dfeq{arvvad}   \\
&& v - v_a -\uh_{tr}\calR(F_r)\in \calN_{[0,F_{rM}]}(F_r) \dfeq{vvad}  \\
&&A_h \uh_{ph}\calR(A_hP_c-F_h) \in\calN_{(-\infty,P_c]}(P)\dfeq{ahphpc} \\
&& Q \in U_b\calR(P) + A_h \uh_{ph}\calR(A_hP_c-F_h)+\calN_{(-\infty,P_M]}(P) \dfeq{qinub}\\
&& f=F_h-F_r \\
&& v_a = \uh_a \max(\calR(F_r-\hat{A}F_h,0)) \dfeq{vau}
\eeqa
\end{subequations}
where $v_a \define Q_a/A_r$, $\uh_a\define c_au_a/A_r^{3/2}$ and $\hat{A}=A_r/A_h$. 

\par

By carefully observing \eq{thesixbregp}, one can see that $v_a>0$ implies $F_r>0$ from \eq{vau}, which implies $v_a\leq v$ from \eq{vvad}. Its contraposition is that $v_a>v$ implies $v_a=0$. Therefore, \eq{thesixbregp} imposes the condition $0\leq v_a\leq\max(0,v)$. When $v<0$, the solutions are obviously $v_a=0$ and $f\in\Gamma(v)$. Therefore, hereafter we consider only the case $v>0$, in which $v_a>0$ may take place. By using the functions defined in \eq{defGammah} and \eq{defGammar}, the first four equations of \eq{thesixbregp} can be rewritten as $F_h\in\Gamma_h(v-\hat{A}v_a)$ and $F_r\in\Gamma_r(v-v_a)$ with $P_c$ and $P$ being eliminated. Considering the definitions of $\Gamma_h$ and $\Gamma_r$, because of the conditions $v-v_a\geq 0$, $v>0$, $u_{tr}>0$ and $\hat{A}\in(0,1)$, $\Gamma_h$ and $\Gamma_r$ are always single-valued and thus can be replaced by $\Gamma_{h+}$ and $\Gamma_{r+}$, respectively. Therefore, \eq{thesixbregp} can be rewritten as follows:
\begin{subequations}\dfeq{dfa}
\beqa
&& \Xi_v(v,v_a)\in -\calN_{[0,\infty)}(v_a)     \dfeq{thesolfva}\\
&& f=\Gamma_{h+}(v-\hat{A}v_a)-\Gamma_{r}(v-v_a)\dfeq{fgamm}
\eeqa
\end{subequations}
where 
\beqa
\Xi_v(v,v_a)&\define& \calS(v_a) - \uh_a^2(\Gamma_{r+}(v-v_a) -\hat{A}\Gamma_{h+}(v-\hat{A}v_a)).
\eeqa
This expression can be seen as an algebraic problem regarding $\{f,v_a\}$ with a given $v$. 

\par

The function $\Xi_v(v,v_a)$ is an increasing function of $v_a$ and it satisfies $\Xi_v(v,v)>0$ because of $\Gamma_{r+}(0)=0$. With the algebraic constraint \eq{thesolfva}, $\Xi_v(v,0)\geq 0$ implies that the solution is $v_a=0$, Meanwhile, if $\Xi_v(v,0)<0$, the solution can be found within the region $v_a\in[0,v]$ by simple root-finding schemes. Once the solution $v_a$ is obtained, $f$ is obtained by \eq{fgamm}.

\par

In conclusions, $f$ and $v_a$ satisfying \eq{dfa} are obtained by the following functions:
\beqa
f\in\Gamma_{\mathrm{reg}}(v)&\define& \lb\{\barr{lr}
\Gamma(v) & \mbox{if } u_a =0\,\vee\,\bmuh\not\in\bmcalU_+\,\vee\,v\leq 0\,\vee\,\Xi_v(v,0)\geq 0\\
\Gamma_{h+}(v-\hat{A}\widehat{v}_a(v))-\Gamma_{r+}(v-\widehat{v}_a(v))  \hspace{-4cm}   & \mbox{otherwise}  \\
\earr\rb.
\\
v_a=\widehat{v}_a(v) &\define&
\lb\{\barr{lr}
0 & \mbox{if } u_a =0\,\vee\,\bmuh\not\in\bmcalU_+\,\vee\,v\leq 0\,\vee\,\Xi_v(v,0)\geq 0 \\
\mbox{FindRoot}(\Xi_v(v,\bullet),[0,v]) \hspace{-2.5cm} & \mbox{otherwise.} 
\earr\rb.
\eeqa
Here, ``$\mbox{FindRoot}$" is a function that finds a root of the argument function within the range specified by the second argument. This computation can be performed with a common iterative method, such as the Bisection method or the false position method, because $\Xi_v(v,v_a)$ is continuous and monotonic. It may also be possible to find analytical methods because $\Xi_v(v,v_a)$ is only a piece-wise parabolic function. The function $\Gamma_{\mathrm{reg}}$ can be seen as an extension of the quasistatic map $\Gamma$ given in Section~\ref{ss:vfrel}.

\par

\begin{figure}[t!]\begin{center}
\includegraphics[scale=0.8]{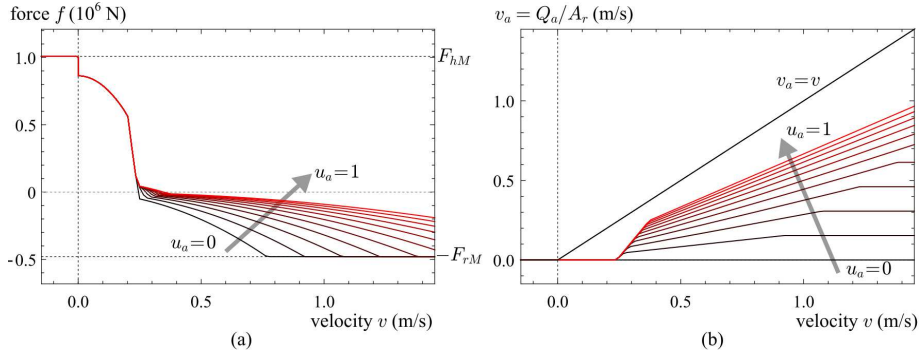}\\[-4pt]
\caption{Numerical examples of the quasistatic map $f\in\Gamma_{\mathrm{reg}}(v)$ and $v_a=\widehat{v}_a(v)$ of the circuit of Fig.~\ref{fg:reg_circuit}, which includes a regeneration pipeline. The parameters are set the same as in Section~\ref{ss:numestt} except $C_a=0.6$ and $a_a=0.0001$~m$^2$. The valve openings are set as $u_c=0.5$ (see~\eq{ucgiven}) and $u_b=0.2$.}\label{fg:reg_nume_stt}
\end{center}\end{figure}

Some numerical examples of $\Gamma_{\mathrm{reg}}$ and $\widehat{v}_a$ are presented in Fig.~\ref{fg:reg_nume_stt}. It can be seen that, as $u_a$ increases (i.e., as the regeneration valve opens), the extending velocity of the rod tends to increase especially in the region $f<0$, i.e., under stretching external forces. It can also be seen that the flowrate $Q_a$ through the regeneration valve increases as $u_a$ and $v$ increase. These features are consistent with those of actual hydraulic actuators.


\subsection{Incorporation into Simulators}\label{ss:regdyn}

As has been discussed in Section~\ref{sc:dyn}, multibody simulations involving the quasistatic map $f\in\Gamma(v)$ require its correspondent $\Lambda$ defined by \eq{deflambda}. In the same manner, now we consider obtaining the function $\Lambda_{\mathrm{reg}}$ that satisfies the following:
\beqa
\beta v+\fb \in \Gamma_{\mathrm{reg}}(v) &\iff & v = \Lambda_{\mathrm{reg}}(\beta,\fb).
\eeqa
To this end, we consider the algebraic problem \eq{dfa} with $f$ being replaced by $f=\beta v+\fb$ and $v$ being treated as an unknown, which is written as follows:
\begin{subequations}\dfeq{casenormalb}
\beqa
\Xi_v(v,v_a) &\in& -\calN_{[0,\infty)}(v_a) \dfeq{caff}\\
\Xi_f(v,v_a) &=& 0 \dfeq{caffb}
\eeqa
\end{subequations}
where 
\beqa
\Xi_f(v,v_a)&\define& \beta v+ \fb - \Gamma_{h+}(v-\hat{A}v_a) + \Gamma_{r+}(v-v_a).
\eeqa
The expression \eq{casenormalb} is an algebraic problem regarding $\{v,v_a\}$ with a given $\fb$. 

\par

The problem \eq{casenormalb} is illustrated in Fig.~\ref{fg:reg_prob}, in which the sets of $\{v,v_a\}$ satisfying \eq{caff} and \eq{caffb} are shown as curves. The functions $\Xi_v(v_a,v)$ and $\Xi_f(v_a,v)$ are increasing with respect to $v_a$ and $v$, respectively. The region $0\leq v_a<v$ should be searched for the solution $\{v,v_a\}$. Because $f=\beta v+\fb \in[-F_{rM},F_{hM}]$, one can see that the solution exists in the region $(-F_{rM}-\fb)/\beta<v<(F_{hM}-\fb)/\beta$. In addition, the solution $v$ of $\Xi_f(v,v_a)=0$ with $v_a>0$ is larger than $\Lambda(\beta,\fb)$, which is the solution with $v_a=0$. Therefore, the trapezoidal area in Fig.~\ref{fg:reg_prob} is the region that must be searched for the solution $\{v,v_a\}$.

\par

As shown in Fig.~\ref{fg:reg_prob}(a), if $\{\Lambda(\beta,\fb),0\}$, which is an analytically obtained initial guess, resides in the region $\Xi_v(v,v_a)<0$, it is the solution of the problem \eq{casenormalb}. Otherwise, the solution can be searched for iteratively, as illustrated in Fig~\ref{fg:reg_prob}(b), by alternately searching in $v_a$ and $v$-directions from the initial guess $\{\Lambda(\beta,\fb),0\}$. The following is an algorithm to obtain the solution:
\begin{subequations}
\beqa
&& \mbox{Function }\Lambda_{\mathrm{reg}}(\beta,\fb) \\ 
&& \quad v := \Lambda(\beta,\fb) \\ 
&& \quad \mbox{If } u_a>0\,\wedge\,\bmuh\in\bmcalU_+\,\wedge\,v > 0 \,\wedge\, \Xi_v(v,0)<0 \  \mbox{Then }\\
&& \quad\quad v_a := 0 \\
&& \quad\quad \mbox{Loop} \\ 
&& \quad\quad \quad v_a := \mbox{FindRoot}\lb( \Xi_v(v,\bullet), [v_a,v]    \rb)\\ 
&& \quad\quad \quad \mbox{If } |\Xi_f(v,v_a)|<\varepsilon_f\ \mbox{Exit Loop} \\
&& \quad\quad \quad v   := \mbox{FindRoot}\lb( \Xi_f(\bullet,v_a), [v,(F_{hM}-\fb)/\beta] \rb)\\ 
&& \quad\quad \quad \mbox{If } |\Xi_v(v,v_a)|<\varepsilon_v\ \mbox{Exit Loop} \\
&& \quad\quad \mbox{While} \\
&& \quad \mbox{End If} \\
&& \mbox{Return } v.
\eeqa
\end{subequations}
Here, $\varepsilon_f$ and $\varepsilon_v$ are very small positive numbers. This algorithm can be seen as an extension of $\Lambda(\beta,\fb)$ given in Section~\ref{ss:Lambda}. As is the case with $\Gamma_{\mathrm{reg}}$, the computation of ``\mbox{FindRoot}" can be performed with common iterative solvers. Some analytical methods may be found because $\Xi_f$ is also a piece-wise parabolic function.

\begin{figure}[t!]\begin{center}
\includegraphics[scale=0.9]{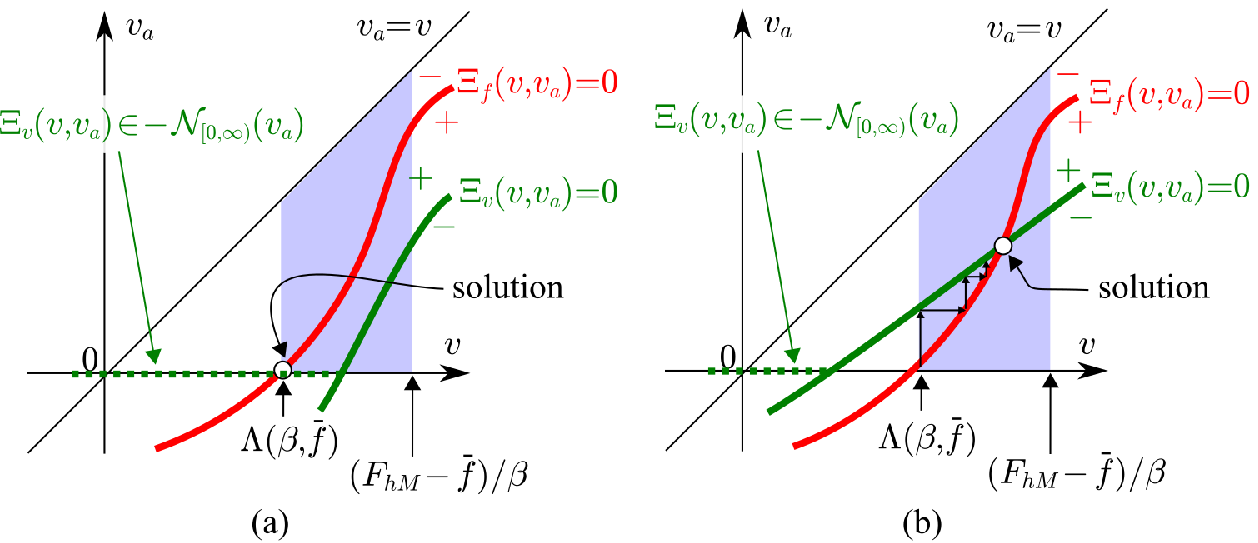}\\[-4pt]
\caption{Solution of the problem \eq{casenormalb} in (a) the case of $\Xi_v(\Lambda(\beta,\fb),0)\geq 0$, where $\{v,v_a\}=\{\Lambda(\beta,\fb),0\}$ is the solution, and (b) the case of $\Xi_v(\Lambda(\beta,\fb),0)<0$, where iterative computation is needed.}\label{fg:reg_prob}
\vspace{20pt}
\includegraphics[scale=0.8]{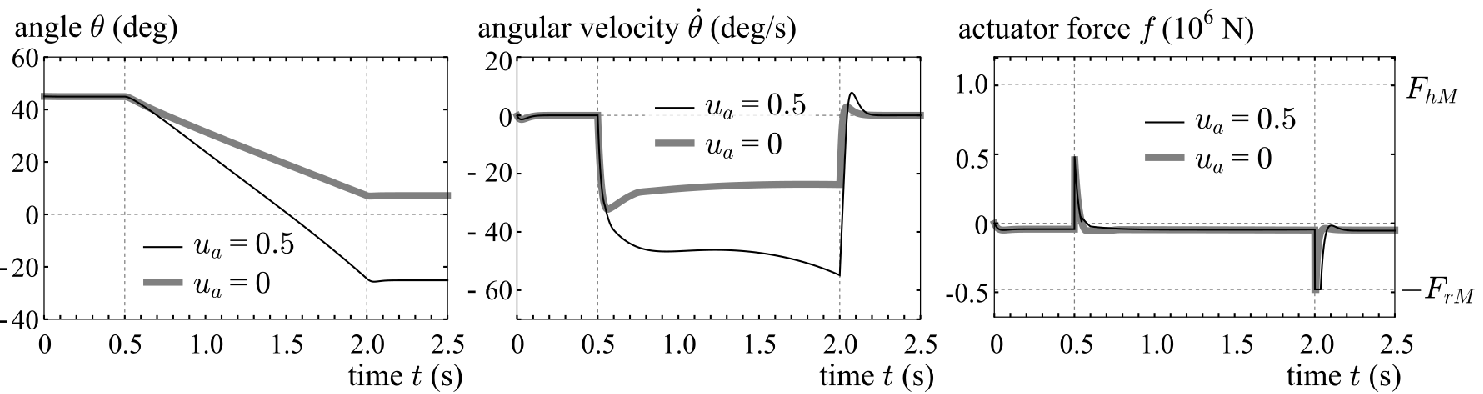}\\[-4pt]
\caption{Simulation results of the system of Fig.~\ref{fg:dynamic_sys} with the regeneration valve being open ($u_a=0.5$) and closed ($u_a=0$).}\label{fg:reg_nume_dyn}
\end{center}\end{figure}

\par

Some simulations were performed with the same arm system as in Section~\ref{ss:numedyn} and Fig.~\ref{fg:dynamic_sys}. The arm was moved down from $\theta=30$~deg for $1.5$~s with the command $u_c=0.5$ with the regeneration valve being open ($u_a=0.5$) and closed ($u_a=0$). The results are shown in Fig.~\ref{fg:reg_nume_dyn}. It shows that $u_a>0$ realizes a faster extending motion, which is consistent with the effects of actual regeneration circuits. In addition, it can be seen that the impulsive actuator force at the beginning of the motion (around $t=0.5$~s) is smaller with $u_a>0$. It can be explained by the fact that the open regeneration valve allows more flow of the oil, resulting in a softer response. The impulses at the end of the movement (around $t=2$~s) should not be directly compared because the angle and the velocity at this time are quite different between the two conditions.

\section{Extension 2: Multiple Actuators Driven by One Pump}\label{sc:mul}
\newcommand\Gammawh{\widehat{\Gamma}}
\newcommand\Lambdawh{\widehat{\Lambda}}
\newcommand\bmbeta{\bm{\beta}}
\newcommand\bmGamma{\bm{\Gamma}}
\newcommand\bmLambda{{\bm{\Lambda}}}

\subsection{Quasistatic Model}\label{ss:mulstt}

In some excavators, more than one actuators are actuated by a single pump as illustrated in Fig.~\ref{fg:mul_circuit}. In such a circuit, the behaviors of the actuators influence each other, e.g., the movement of one actuator may decrease the supplied flowrate to other actuators. This section presents an extension of the quasistatic model to deal with such systems.

\par

In circuits like Fig.~\ref{fg:mul_circuit}, the actuators are connected to a junction at which the pressure is $P$, and the total supplied flowrate $Q$ from the pump is equal to the sum of the supplied flowrates to all actuators plus that discharged through the bleed valve and the pump relief valve. On the other hand, the actuator model developed in Section~\ref{ss:vfrel} assumes that the supplied flowrate $Q$ from the pump is the input to be given. Therefore, one can see that it is convenient to have a modified version of the quasistatic map $\Gamma$ of which the input is the pressure $P$ instead of the flowrate $Q$.

\begin{figure}[t!]\begin{center}
\includegraphics[scale=1.2]{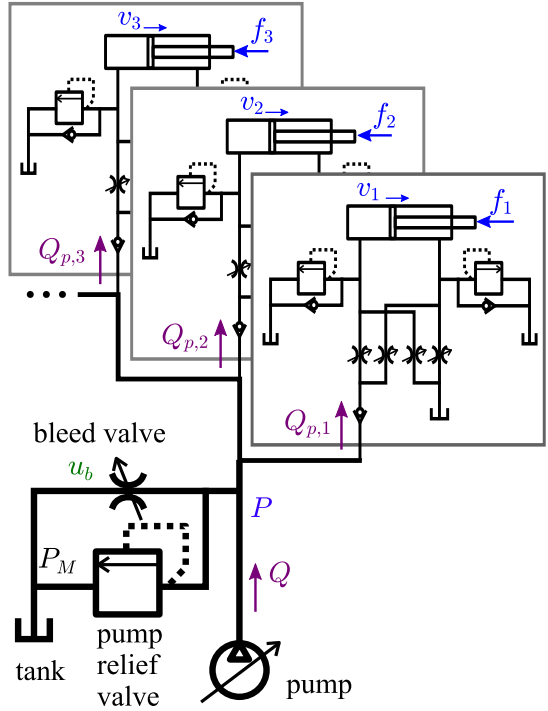}\\[-4pt]
\caption{Multiple actuators driven by one pump.}\label{fg:mul_circuit}
\end{center}\end{figure}

With a close look at the definition \eq{Gammapp} of the function $\Gamma$ in Section~\ref{ss:vfrel}, one can see that a $P$-input version of the quasistatic map $\Gamma$ can be obtained by replacing $P_m$ by $P$ and removing the segments depending on $Q$. Specifically, the modified version $\Gammawh(P,v)$, which maps the velocity $v$ to the force $f$ of an actuator depending on $P$, can be given as follows:
\begin{subequations}\dfeq{Gammappbb}
\beqa
f&\in&\Gammawh(P,v)\define\gsgn(\Gammawh_-(P,v),v,\Gammawh_+(P,v)) \dfeq{fGammawhPv}
\eeqa
where
\beqa
\Gammawh_+(P,v)&\define&\max(\min( \max(\Gamma_{+0a}(v),\Gamma_{+0b}(v)),
\non\\&&\quad
 \max(\Gammawh_{+2a}(P,v),\Gammawh_{+2b}(P,v))),\Gamma_{+3}(v),-F_{rM}) \dfeq{defgammaposm} \quad\\
\Gammawh_-(P,v)&\define&\min(\max(\min(\Gamma_{-0a}(v),\Gamma_{-0b}(v)),\non\\
&&\quad \min(\Gammawh_{-2a}(P,v),\Gammawh_{-2b}(P,v))),\Gamma_{-3}(v), F_{hM}) \quad\\
\Gammawh_{+2a}(P,v)&\define&A_hP-\dfrac{\calS(v)}{\uh_{ph}^2} -\dfrac{\calS(v)}{\uh_{tr}^2} \\
\Gammawh_{+2b}(P,v)&\define&A_hP-\dfrac{\calS(v)}{\uh_{ph}^2} -F_{rM} \\
\Gammawh_{-2a}(P,v)&\define&-A_rP-\dfrac{\calS(v)}{\uh_{pr}^2}-\dfrac{\calS(v)}{\uh_{th}^2} \\
\Gammawh_{-2b}(P,v)&\define&-A_rP-\dfrac{\calS(v)}{\uh_{pr}^2}+F_{hM}.
\eeqa
\end{subequations}
Here, the functions without hats are those defined in \eq{Gammapp}. In the same manner as \eq{gammazero}, we also have the following:
\begin{subequations}\dfeq{gammazeroP}
\beqa
\Gammawh_+(P,0) =\Gammawh_{h+}(P,0)-\Gammawh_{r+}(P,0),\quad \Gammawh_-(P,0)=\Gammawh_{h-}(P,0)-\Gammawh_{r-}(P,0)
\eeqa
where
\beqa
\Gammawh_{h+}(P,0) &=&\lb\{\barr{ll}
\min(F_{hM},A_hP) &\mbox{if }  \uh_{ph}> 0   \\
0 &\mbox{if } \uh_{ph}=0  \\
\earr\rb.
\\
\Gammawh_{r+}(P,0) &=&\lb\{\barr{ll}
0      &\mbox{if }  \uh_{tr}> 0  \\
F_{rM} &\mbox{if }    \uh_{tr}=0  \\
\earr\rb.
\\
\Gammawh_{h-}(P,0) &=&\lb\{\barr{ll}
0      &\mbox{if }  \uh_{th}>0  \\
F_{hM} &\mbox{if }  \uh_{th}=0 
\earr\rb. 
\\
\Gammawh_{r-}(P,0) &=&\lb\{\barr{ll}
\min(F_{rM},A_rP_M) &\mbox{if } \uh_{pr}>0 \\
0                   &\mbox{if } \uh_{pr}=0. \\
\earr\rb. 
\eeqa
\end{subequations}

\par

The flowrate $Q_p$ into an actuator is determined by the pressure $P$ at the junction and the pressure $P_h$ or $P_r$ of the chamber connected to the pump, specifically, as follows:
\beqa
Q_p = A_h \uh_{ph}\max(\calR(A_hP-F_h),0) + A_r \uh_{pr}\max(\calR(A_rP-F_r),0).
\eeqa
If $\bmuh\in\bmcalU_+$, we have $F_h=f+F_r$ and $F_r = \proj_{[0,F_{rM}]} ( \calS(v)/\uh_{tr}^2 )$. If $\bmuh\in\bmcalU_-$, we have $F_h = \proj_{[0,F_{hM}]}(\calS(v)/\uh_{th}^2)$ and $F_r = F_h - f$. Therefore, $Q_p$ is obtained as follows:
\beqa
Q_p = \widehat{Q}_p(P,v,\Gammawh(P,v)) \dfeq{Qpgmm}
\eeqa
where
\beqa
\widehat{Q}_p(P,v,\Gammawh(P,v))
 &\define& \lb\{\barr{ll}
 A_h \uh_{ph}\max(\calR(A_hP-\proj_{[0,F_{rM}]}( \calS(v)/\uh_{tr}^2 )-f),0) &\mbox{if }\bmuh\in\bmcalU_+ \\
 0                                                                           &\mbox{if }\bmuh\in\bmcalU_0 \\
 A_r \uh_{pr}\max(\calR(A_rP-\proj_{[0,F_{hM}]}(-\calS(v)/\uh_{th}^2 )+f),0) &\mbox{if }\bmuh\in\bmcalU_-. \\
\earr\rb.
\dfeq{qpah}
\eeqa
The function $\Gammawh$ appearing in \eq{Qpgmm} may be set-valued at $v=0$, but a careful observation of its limits of both sides of zero shows that $\widehat{Q}_p(P,0,\Gammawh(P,0))=0$, single-valued, under all three conditions, $\bmcalU_+$, $\bmcalU_0$ and $\bmcalU_-$.

\par

By using the function $\widehat{Q}_p$, the algebraic constraint between the total flowrate $Q$ from the pump and the pressure $P$ can be described as follows:
\begin{subequations}\dfeq{Qttvf}
\beqa
\Xi_P(P) \in   \calN_{(-\infty,P_M]}(P) \dfeq{Qttvfinq}
\eeqa
where
\beqa
\Xi_P(P) &\define& Q - U_b\calR(P)  - \sum_{j=1}^N \widehat{Q}_{p,j}(P,v_j,\Gammawh_j(P,v_j)).
\eeqa
\end{subequations}
Here, the symbols with the subscript $j$ stand for those associated with the $j$th actuator, and $N$ denotes the number of actuators. The value of $\Xi_P(P)$ can be interpreted as the flowrate from the pump relief valve (see Fig.~\ref{fg:mul_circuit}), which is the difference between the oil supply from the pump and the sum of the oil supplies to all actuators plus that to the bleed valve. As long as $P<P_M$, \eq{Qttvfinq} reduces to $\Xi_P(P)=0$, which means that the pump relief valve is closed. When $P=P_M$, \eq{Qttvfinq} reduces to $\Xi_P(P)\geq 0$, which means that the oil is discharged from the pump relief valve, of which the pressure limit is $P_M$. The pressure $P$ can be found by solving the algebraic problem \eq{Qttvf} as follows:
\beqa
P&=& \lb\{\barr{ll} P_M &\mbox{if }  \Xi_P(P_M)\geq 0  \\ 
\mbox{FindRoot}(\Xi_P(\bullet),[0,P_M])  &\mbox{otherwise.}\earr\rb. \dfeq{Pobt}
\eeqa
This ``FindRoot" is also easy because of the monotonicity of the function $\Xi_P$.

\par

By using the value of $P$ obtained by \eq{Pobt}, the quasistatic relation between the forces $\bmf=[f_1,\cdots,f_N]^T$ and the velocities $\bmv=[v_1,\cdots,v_N]^T$ of $N$ actuators are written in the following form: 
\beqa
\bmf \in \bmGamma_{\mathrm{mul}}(\bmv) \define \lb[ \Gammawh_1(P,v_1), \cdots, \Gammawh_N(P,v_N)\rb]^T \dfeq{gmmwh}
\eeqa
where $P$ is the one obtained by \eq{Pobt}.

\par

Some numerical examples are shown in Fig.~\ref{fg:mul_nume_stt}. In these examples, two identical actuators 1 and 2 share a single pump. The parameters of the actuators are the same as those in Section~\ref{ss:numestt}. The force $f_1$ obtained by the map $[f_1,f_2]^T\in\bmGamma_{\mathrm{mul}}([v_1,v_2]^T)$ according to the variable $v_1$ and some fixed values of $v_2$ are shown in Figs.~\ref{fg:mul_nume_stt}(a) and (d). Intermediate values $P$ and $\Xi_P(P)$, which are the immediate output of the root finding in \eq{Pobt}, are presented in Figs.~\ref{fg:mul_nume_stt}(b)(c) and (e)(f). It can be seen that an increased speed $v_2$ of the actuator 2 results in a decreased speed $v_1$ of the actuator 1, a decreased pressure $P$ at the junction, and a decreased flowrate $\Xi_P(P)$ from the pump relief valve, which are consistent with what can happen in real hydraulic circuits.

\begin{figure}[t!]\begin{center}
\includegraphics[scale=0.8]{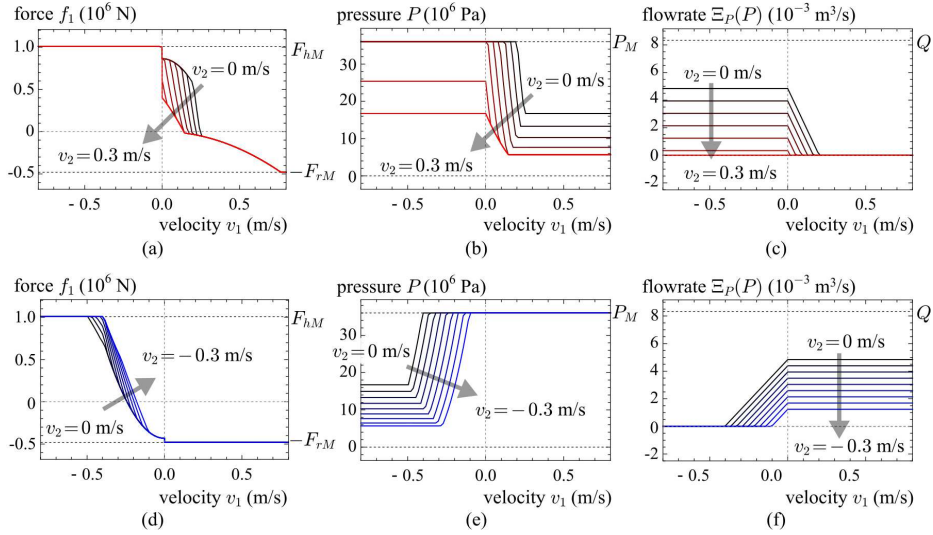}\\[-4pt]
\caption{Numerical examples regarding the quasistatic map $[f_1,f_2]^T\in\bmGamma_{\mathrm{mul}}([v_1,v_2]^T)$ of a circuit including two identical actuators driven by a single pump. Positive commands $u_{c,1}=u_{c,2}=0.5$ are given in (a)-(c) and negative commands $u_{c,1}=u_{c,2}=-0.5$ are given in (d)-(f) (see \eq{ucgiven}).  }\label{fg:mul_nume_stt}
\end{center}\end{figure}

\subsection{Incorporation into Simulators}\label{ss:muldyn}

In the same way as in Section~\ref{sc:dyn} and in Section~\ref{ss:regdyn}, the application of the quasistatic map $\bmGamma_{\mathrm{mul}}$ to multibody simulations requires a map $\bmLambda_{\mathrm{mul}}$ that satisfies the following relation:
\beqa
\diag{[}\bmbeta]\bmv+\bmfb\in\bmGamma_{\mathrm{mul}}(\bmv) &\iff& \bmv \in\bmLambda_{\mathrm{mul}}(\bmbeta,\bmfb) \dfeq{lbmlkfe00}
\eeqa
where $\bmbeta$, $\bmv$ and $\bmfb$ are $N$-dimenional vectors and the elements of $\bmbeta$ are all positive. The structure \eq{gmmwh} of $\bmGamma_{\mathrm{mul}}$ suggests that $\bmLambda_{\mathrm{mul}}$ has the following structure:
\beqa
\bmLambda_{\mathrm{mul}}(\bmbeta,\bmfb) = [ \Lambdawh_1(P,\beta_1,\fb_1), \cdots, \Lambdawh_N(P,\beta_N,\fb_N)]^T \dfeq{lbmdfe}
\eeqa
where $\Lambdawh_j(P,\beta_j,\fb_j)$ are those that satisfy the following:
\beqa
\beta_j v_j + \fb_j \in\Gammawh_j(P,v_j) \iff v_j\in\Lambdawh(P,\beta_j,\fb). \dfeq{deflambdam}
\eeqa
Recalling that $\Lambda$ is obtained as \eq{defvvvv} and \eq{defLambdadetail}, one can obtain $\Lambdawh_j$ as follows:
\begin{subequations}\dfeq{Lambdawhdef}
\beqa
\Lambdawh(P,\beta,\fb) 
&=&
 \lb\{\barr{ll}
 v_{hM} & \mbox{if } \bar{f}\geq F_{hM}  \wedge \bmu\in\bmcalU_+ \\
 \max(\min(\max(v_{+0a},v_{+0b}), \max(\widehat{v}_{+2a},\widehat{v}_{+2b})),v_{+3},v_{rM}) \hspace{-5cm}&\\& \mbox{if } \bar{f}\leq  \Gammawh_+(P,0) \wedge \bmu\in\bmcalU_+ \\
 v_{rM} & \mbox{if } \fb\leq -F_{rM} \  \wedge\  \bmu\in\bmcalU_-\\
 \min(\max(\min(v_{-0a},v_{-0b}),  \min(\widehat{v}_{-2a},\widehat{v}_{-2b})),v_{-3},v_{hM}) \hspace{-5cm}&\\& \mbox{if } \fb\geq\Gammawh_-(P,0)\  \wedge\  \bmu\in\bmcalU_- \\
 0   &\mbox{if } \Gammawh_+(P,0)<\bar{f}<\Gammawh_-(P,0)
\earr\rb. 
\eeqa
where
\beqa
\widehat{v}_{+2a}&\define& \Phi_A\lb(\beta ,\fb-A_hP, \psi(\uh_{ph},\uh_{tr})\rb)\\
\widehat{v}_{+2b}&\define& \Phi_A\lb(\beta ,\fb-A_hP+F_{rM}, \uh_{ph}^2\rb)\\      
\widehat{v}_{-2a}&\define& \Phi_A\lb(\beta,\fb+A_rP       ,\psi(\uh_{pr},\uh_{th})\rb)\\ 
\widehat{v}_{-2b}&\define& \Phi_A\lb(\beta,\fb+A_rP-F_{hM},\uh_{pr}^2             \rb).
\eeqa
\end{subequations}
Here, the subscript $j$ denoting the actuator indices are omitted and the symbols without hats are those defined in \eq{defvvvv}.

\par

\begin{figure}[t!]\begin{center}
\includegraphics[scale=0.8]{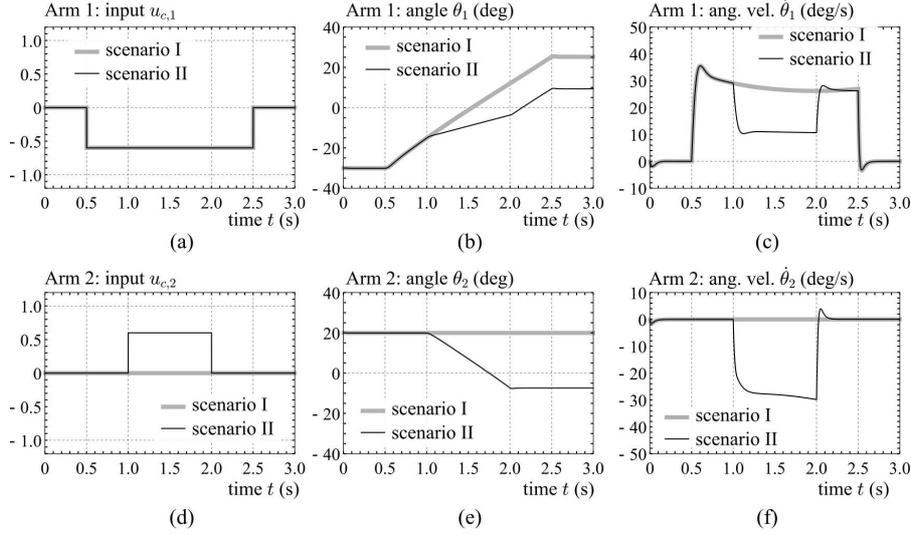}\\[-4pt]
\caption{Simulation results of two arms identical to the one in Fig.~\ref{fg:dynamic_sys} driven by a single pump. No external forces were applied to either arm. The command $u_{c,1}$ sent to the arm~1 was the same between the two scenarios, while the arm 2 was actuated only in Scenario~\RmN{2}. }\label{fg:mul_nume_dyn}
\end{center}\end{figure}

In the expression \eq{lbmdfe}, the pressure $P$ can be obtained in the same line of thought as in Section~\ref{ss:mulstt}, but the problem \eq{Qttvf} needs to be modified because $v_j$ are unknowns. Considering that $\bmfb$ is given, the problem is redefined as
\begin{subequations}\dfeq{Qtt}
\beqa
\Xi_{P\Lambda}(P)  \in \calN_{(-\infty,P_M]}(P) 
\eeqa
where
\beqa
\Xi_{P\Lambda}(P) &\define& Q - U_b\calR(P) - \sum_{j=1}^N \widehat{Q}_{p,j}(P,\Lambdawh_j(\beta_j,\fb_j),\fb_j+\beta_j \Lambdawh_j(\beta_j,\fb_j)).
\eeqa
\end{subequations}
The solution $P$ of the problem \eq{Qtt} can be written as follows:
\beqa
P&=& \lb\{\barr{ll} P_M &\mbox{if }  \Xi_{P\Lambda}(P_M)\geq 0  \\ 
\mbox{FindRoot}(\Xi_{P\Lambda}(\bullet),[0,P_M])  &\mbox{otherwise.}\earr\rb. \dfeq{fdafeoo}
\eeqa
In conclusion, the function $\bmLambda_{\mathrm{mul}}$ satisfying \eq{lbmlkfe00} is obtained as \eq{lbmdfe} in which $\Lambdawh_j$ are those defined by \eq{Lambdawhdef} and $P$ is replaced by \eq{fdafeoo}.

\par

Some simulations were performed with a system composed of two arms identical to the one in Fig.~\ref{fg:dynamic_sys} connected with a single pump. It was implemented with the function $\bmLambda_{\mathrm{mul}}$ with $N=2$. The parameters of the arms were set the same as in Section~\ref{ss:numedyn}, and the parameters of the actuators and the circuit were set the same as in Section~\ref{ss:mulstt}. No external forces were applied to either arm. The results are shown in Fig.~\ref{fg:mul_nume_dyn}. The initial postures of the arms 1 and 2 were set as $\theta=-30$~deg and $\theta=20$~deg, respectively. In the first scenario, the arm 1 was driven upward by the command $u_{c,1}=-0.6$ from $t=0.5$~s to $2.5$~s. while the arm 2 was not driven. In the second scenario, the command to the arm 1 was the same as that in the first scenario, the arm 2 was driven downward from $t=1$~s to $2$~s. As can be seen in Figs.~\ref{fg:mul_nume_dyn}(b) and (c), the actuation of the arm 2 affected the motion of the arm 1, i.e., the arm 1 was decelerated when the arm 2 was actuated. It can be explained by the effect of the arm 2, which suctioned the oil when actuated, causing the shortage of the oil flow into the arm 1.

\section{Conclusions}\label{sc:ccl}

This article has presented a quasistatic model of a hydraulic actuator driven by a four-valve independent metering circuit. The presented model is described as a nonsmooth map between the velocity and the force. The model is derived from the algebraic constraint between the flowrate and the pressure at every valve in the circuit in the steady state. This article also presents an approach to incorporate the quasistatic model into the multibody dynamics simulators, in which the hydraulic model is connected to rigid bodies through virtual viscoelastic elements. In addition, the proposed model is extended to include a regeneration pipeline and to deal with a collection of actuators driven by a single pump.

\par

In multibody simulations employing the presented quasistatic model, the transient responses are determined by the virtual viscoelastic elements, which are parameterized by the stiffness $K$ and the viscosity $B$. In reality, transient responses are governed by many factors such as the compressibility of the oil, the inertia of the actuator and the oil, and the compliance of the pipes. Some tuning guidelines for the parameters $K$ and $B$ should be sought to reproduce the behaviors of actual systems by the presented simulation framework. In order to accelerate the computation, analytical methods for the root-finding routines that have appeared in the proposed algorithm in Sections~\ref{sc:reg} and \ref{sc:mul} should also be addressed. Integration of the schemes of Sections~\ref{sc:reg} and \ref{sc:mul}, i.e., multiple actuators with regeneration pipelines driven by a single pump, is also an open problem.


\end{document}